\documentclass[letterpaper,UKenglish,cleveref, autoref,thm-restate]{lipics-v2021}

\pdfoutput=1 
\hideLIPIcs  

\graphicspath{{./figures/}}

\title{Relating Interleaving and Fréchet Distances via Ordered Merge Trees}

\author{Thijs Beurskens}{Department of Mathematics and Computer Science, TU Eindhoven, The Netherlands}{t.p.j.beurskens@tue.nl}{}{Supported by the Dutch Research Council (NWO) under project no. OCENW.M20.089.}

\author{Tim Ophelders}{Department of Information and Computing Science, Utrecht University, The Netherlands \and Department of Mathematics and Computer Science, TU Eindhoven, The Netherlands}{t.a.e.ophelders@uu.nl}{https://orcid.org/0000-0002-9570-024X}{Supported by the Dutch Research Council (NWO) under project no. VI.Veni.212.260.}

\author{Bettina Speckmann}{Department of Mathematics and Computer Science, TU Eindhoven, The Netherlands}{b.speckmann@tue.nl}{https://orcid.org/0000-0002-8514-7858}{}

\author{Kevin Verbeek}{Department of Mathematics and Computer Science, TU Eindhoven, The Netherlands}{k.a.b.verbeek@tue.nl}{https://orcid.org/0000-0003-3052-4844}{}
\authorrunning{T. Beurskens, T. Ophelders, B. Speckmann, K. Verbeek} 
\Copyright{Thijs Beurskens, Tim Ophelders, Bettina Speckmann and Kevin Verbeek}

\ccsdesc[100]{Theory of computation~Computational geometry}

\keywords{Monotone Interleaving Distance, Merge Tree, Fr\'{e}chet distance}
\category{}

\relatedversion{}
\nolinenumbers 

\EventEditors{John Q. Open and Joan R. Access}
\EventNoEds{2}
\EventLongTitle{42nd Conference on Very Important Topics (CVIT 2016)}
\EventShortTitle{CVIT 2016}
\EventAcronym{CVIT}
\EventYear{2016}
\EventDate{December 24--27, 2016}
\EventLocation{Little Whinging, United Kingdom}
\EventLogo{}
\SeriesVolume{42}
\ArticleNo{23}

\RequirePackage{multicol}
\RequirePackage{mathtools}
\RequirePackage{subcaption}

\newcommand{\enumfig}[1]{{\sffamily\bfseries\upshape\mathversion{bold}#1}}
\newcommand{\enumit}[1]{\textcolor{darkgray}{\sffamily\bfseries\upshape\mathversion{bold}#1}}

\newcommand{\R}{\mathbb{R}}
\newcommand{\bigO}{\mathcal{O}}
\newcommand{\norm}[1]{\left\lVert #1 \right\rVert_{\infty}}

\newcommand{\lca}{\mathrm{lca}}
\newcommand{\im}{\mathrm{Im}}
\newcommand{\depth}{\mathrm{depth}}

\newcommand{\intdist}[1]{d_\mathrm{I}^{\mathrm{#1}}}
\newcommand{\orddist}[1]{d_\mathrm{MI}^{\mathrm{#1}}}
\newcommand{\labdist}{d^\mathrm{L}}
\newcommand{\frechet}{d_\mathrm{F}}

\newcommand{\fl}{f_\ell}

\newcommand{\lei}{\sqsubseteq}

\newcommand{\lel}{\lei_L}

\newcommand{\leafo}{\sqsubseteq_L}
\newcommand{\ileaf}{\sqsubseteq_L}
\newcommand{\sileaf}{\sqsubset_L}
\newcommand{\irel}{\sqsubseteq}
\newcommand{\sirel}{\sqsubset}

\newcommand{\ihat}{{\hat{\imath}}}
\newcommand{\jhat}{{\hat{\jmath}}}

\newcommand{\geom}{T}
\newcommand{\imat}{M}

\newcommand{\ls}[1]{f^{-1}({#1})} 		
\newcommand{\an}[2]{#1|^{#2}}		
\newcommand{\lowa}[1]{#1^F}
\newcommand{\elll}{\hat{\ell}}
\newcommand{\leaf}{u}

\newcommand{\mat}{\mathcal{M}}

\newtheorem{definition2}{Definition}

\usepackage{tikz-cd}
\usepackage{mathtools}

\newcommand{\Img}{\mathrm{Im}}
\newcommand{\id}{\mathrm{id}}
\newcommand{\from}{\colon}
\newcommand{\U}{\mathcal{U}}
\newcommand{\C}{\mathcal{C}}

\newcommand{\Rpos}{\R_{\leq}}
\newcommand{\inv}{^{-1}}
\newcommand{\into}{\hookrightarrow}
\newcommand{\To}{\Rightarrow}

\newcommand{\subinv}[1]{\hat#1^{-1}}

\begin{document}

\maketitle

\begin{abstract}
\noindent Merge trees are a common topological descriptor for data with a hierarchical component, such as terrains and scalar fields.
The interleaving distance, in turn, is a common distance for comparing merge trees.
However, the interleaving distance for merge trees is solely based on the hierarchical structure, and disregards any other geometrical or topological properties that might be present in the underlying data. Furthermore, the interleaving distance is NP-hard to compute.

In this paper, we introduce a form of ordered merge trees that can capture intrinsic order present in the data.
We further define a natural variant of the interleaving distance, the \emph{monotone interleaving distance}, which is an order-preserving distance for ordered merge trees.
Analogously to the regular interleaving distance for merge trees, we show that the monotone variant has three equivalent definitions in terms of two maps, a single map, or a labelling.
Furthermore, we establish a connection between the monotone interleaving distance of ordered merge trees and the Fr\'{e}chet distance of 1D curves. As a result, the monotone interleaving distance between two ordered merge trees can be computed exactly in near-quadratic time in their complexity. 
The connection between the monotone interleaving distance and the Fr\'{e}chet distance builds a new bridge between the fields of topological data analysis, where interleaving distances are a common tool, and computational geometry, where Fr\'{e}chet distances are studied extensively.
\end{abstract}

\section{Introduction}\label{sec:introduction}
Topological descriptors, such as persistence diagrams, Reeb graphs, Morse-Smale complexes, and merge trees, lie at the heart of topological data analysis; they support the analysis and visualisation of complex, real-world data.
In this paper, we focus in particular on merge trees.
Merge trees are frequently defined on a scalar field, where they capture connected components of sublevel sets; they are generally used to represent hierarchical structures.
Merge trees closely resemble phylogenetic trees in computational biology, and dendrograms in hierarchical clustering.

\begin{figure}
    \centering
    \includegraphics[width=0.32\textwidth]{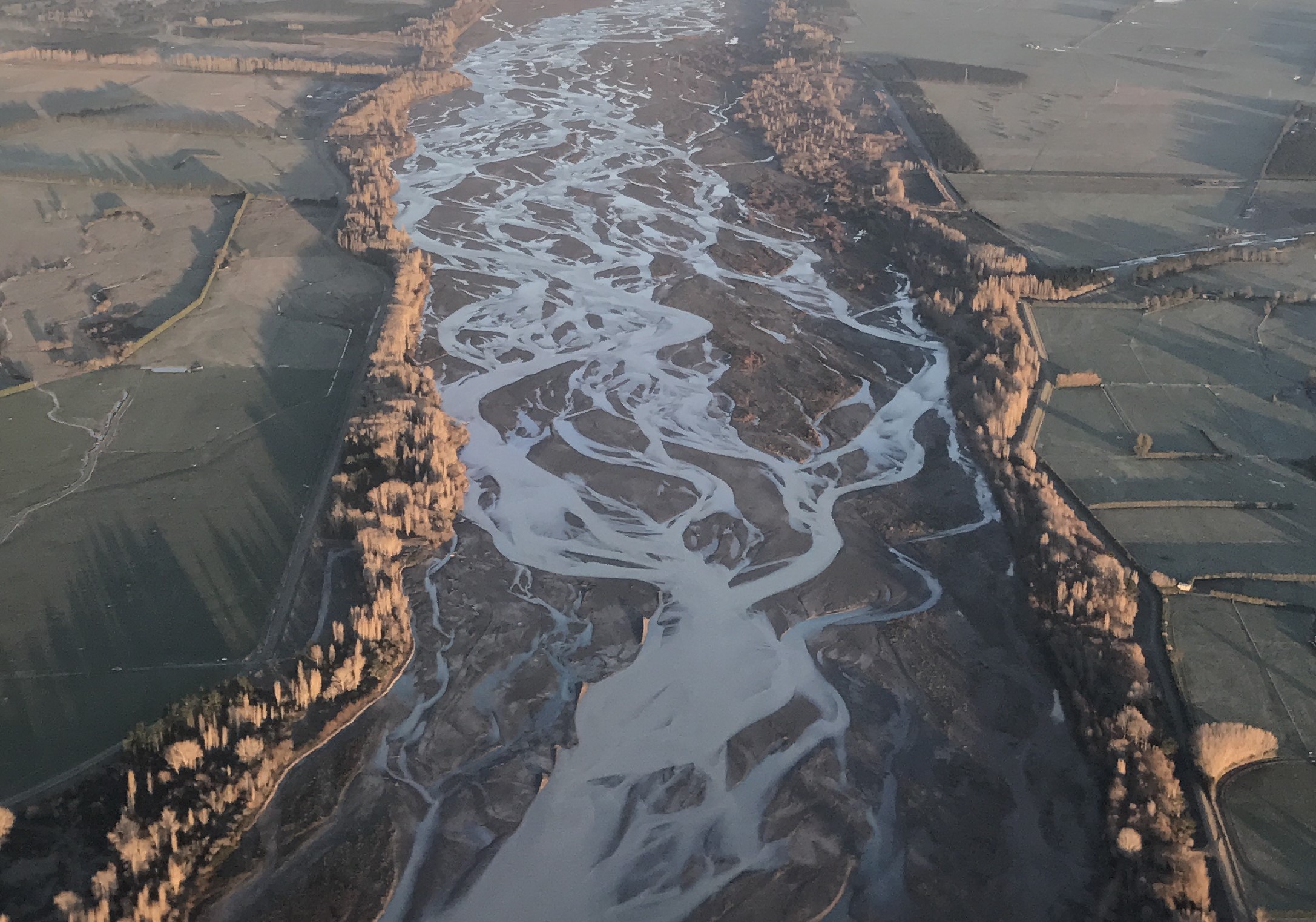}
    \caption{The Waimakariri River in New Zealand is a braided river. Photo by Greg O'Beirne \protect\cite{obeirne2017waimakariri}.}
    \label{fig:braided-river}
\end{figure}

Standard merge trees focus solely on the hierarchy and do not represent other salient geometric features.
However, there are applications where there is an intrinsic order present in the data which we would like to capture in the merge tree as well.
Our work, specifically, is motivated  by the study of braided rivers (see Figure~\ref{fig:braided-river}).
A braided river is a multi-channel river system, known to evolve rapidly~\cite{howard1970topological, marra2014network}. It is our goal to analyse the evolution of such a river system over time. 
Since the water level in a river, and hence the number of recognisable channels, is highly variable, geomorphologists frequently model river networks based solely on bathymetry, that is, terrain models of the river bed~\cite{Hiatt2020JGR, kleinhans2019networks}.
A river network is then characterised by the topological features of this river bed, most notably, by the so-called \emph{braid bars} that correspond to local maxima of the terrain.
For the purpose of studying river networks over time, we use a simple topological descriptor, namely a merge tree based on these bars: each leaf in our merge tree represents a local maximum (bar), and each internal vertex represents two bars merging as we move downward in the terrain (see Figure~\ref{fig:merge-tree}).
We use these river networks and merge trees in the software package TopoTide~\cite{sonke2019topotide} that we have been developing with domain scientists.

\begin{figure}[b]
    \centering
    \includegraphics{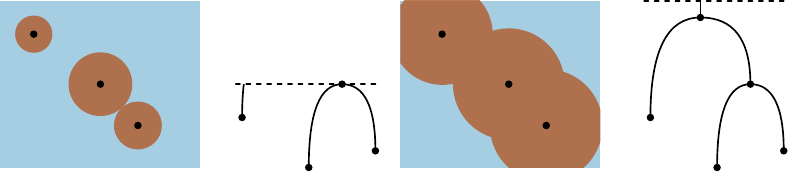}
    \caption{A schematic river, at different levels of the terrain. 
    We use a merge tree to represent the bars: leaves correspond to local maxima, and internal vertices correspond to two bars merging.
    We draw the tree upside down, following the convention that merge trees are rooted at positive infinity.}
    \label{fig:merge-tree}
\end{figure}

To track the evolution of a river over time via a representation with merge trees, we need a geomorphologically meaningful way to match consecutive merge trees while respecting the intrinsic order on the bars induced by the river (upstream towards downstream and from bank to bank).
Hence we need to encode such orders into our merge trees and then match these merge trees in a way that respects the orders.
The standard way to compare and match two merge trees is the interleaving distance~\cite{morozov2013interleaving}.
However, the regular interleaving distance is unable to capture orders and there is no known efficient algorithm to compute even an constant-factor approximation of the interleaving distance, rendering it effectively useless in practice.
In fact, computing the interleaving distance is NP-hard~\cite{agarwal2018computing}.\footnote{Agarwal et al.~\cite{agarwal2018computing} actually prove that approximating the Gromov-Hausdorff distance with a factor better than 3 is NP-hard. As many have observed, this proof also applies to the interleaving distance. 
For completeness, we present the modified proof in Appendix~\ref{app:np-hardness}.
}
We hence introduce a class of ordered merge trees that can capture the intrinsic order present among river bars and an accompanying \emph{monotone inter\-leaving distance} which preserves these orders.

\subparagraph{Contributions.}
We introduce the \emph{monotone interleaving distance}, an extension of the interleaving distance for ordered merge trees based on order-preserving maps.
We show that an ordered merge tree induces a 1D curve and use this to compute the monotone interleaving distance efficiently.
In particular, we show that the monotone interleaving distance between two ordered merge trees is equal to the Fréchet distance between their induced 1D curves.
Hence, we can use any algorithm for the Fréchet distance between 1D curves to compute the monotone interleaving distance between ordered merge trees.
Furthermore, we show that the Fréchet distance, and therefore also the monotone interleaving distance, is an interleaving distance for a specific category.
Lastly, we give two alternative definitions for the monotone interleaving distance, in terms of $\delta$-good maps and in terms of labelled merge trees.
This last result is analogous to a similar result for the regular interleaving distance for merge trees.

In Section~\ref{sec:preliminaries} we discuss the necessary background from computational topology.
We introduce our class of ordered merge trees in Section~\ref{sec:total-orders} and show how to define them in two equivalent ways.
Here, we also define the monotone interleaving distance.
Section~\ref{sec:frechet} focuses on the relation between the Fréchet distance and the monotone interleaving distance.
In Section~\ref{sec:categorical} we switch to a categorical point of view, and we show that the Fréchet distance is an interleaving distance for a specific category.
Lastly, in Section~\ref{sec:variants} we discuss the two alternative definitions: \emph{monotone $\delta$-good interleaving distance} and the \emph{monotone label interleaving distance}.

\subparagraph{Related work.}
The interleaving distance is a distance originally defined on persistence modules, and was first introduced by Chazal et al.~\cite{chazal2009proximity}.
The distance has since been well-studied from a categorical point of view~\cite{bjerkevik2020computing, bjerkevik2018computational, bubenik2015metrics, bubenik2014categorification, curry2022decorated, silva2016categorified, silva2018theory,lesnick2015theory,patel2018generalized}, and has been transferred to numerous topological descriptors such as Reeb graphs~\cite{bauer2014measuring} and homotopy groups~\cite{memoli2019persistent}.
Only few interleaving distances have been shown to be polynomial time computable, such as the interleaving distance for 1-parameter persistence modules~\cite{bjerkevik2018computational} and the interleaving distance for phylogenetic trees~\cite{munch2019the}.
Morozov et al.~\cite{morozov2013interleaving} defined the interleaving distance for merge trees.
Agarwal et al.~\cite{agarwal2018computing} established a relation between the interleaving distance and the Gromov-Hausdorff distance.$^1$
Two alternative definitions have since been formulated: first by Touli and Wang~\cite{touli2022fpt}, who used their result to design an FPT-algorithm for computing the interleaving distance, and secondly by Gasparovich et al.~\cite{gasparovich2019intrinsic}.
Their result has been used to design heuristic algorithms for computing geometry aware labellings~\cite{yan2022geometry, yan2020structural}.

The Fréchet distance is a well-studied distance in the field of computational geometry.
The first algorithm to compute the (continuous) Fréchet distance between polygonal curves was given by Godau~\cite{godau91natural}, who described a $\bigO(n^3\log n)$ time algorithm, where $n$ is the number of vertices of the input curves.
Later, Alt and Godau~\cite{alt95computing} improved this result to a $\bigO(n^2\log n)$ time algorithm, and, in the word RAM model of computation, the complexity was further improved to $\bigO(n^2(\log\log n)^2)$~\cite{buchin17four}.
Moreover, for general curves it is known that no strongly subquadratic algorithm exists, unless the strong exponential time hypothesis (SETH) fails (in dimension $d \ge 2$~\cite{bringmann14hardness} and in dimension $d = 1$~\cite{buchin19seth}).
For restricted classes of 1D curves, faster algorithms do exist~\cite{blank2024faster, buchin2017folding}.

\section{Preliminaries}\label{sec:preliminaries}
In this section, we review the necessary mathematical notions in computational topology and geometry, including merge trees, the interleaving distance, and the Fréchet distance.


\subparagraph{Merge trees.}
A \emph{rooted tree} is a tree where one vertex is identified as the \emph{root}.
Let $T$ denote a rooted tree, with vertices $V(T)$, edges $E(T)$, and leaves $L(T) \subseteq V(T)$.
In the remainder of the paper, we view $T$ not necessarily as a combinatorial tree, but rather as a topological space.
We note here that the distances we consider are defined not only on the vertices of the tree, but also on the interior of the edges.
To that end, we use a continuous representation of $T$, which is commonly referred to as the \emph{geometric realisation}, or, more accurately, the \emph{topological realisation} of $T$.
Formally, the topological realisation of $T$ is a topological space, which is defined as the union of segments that each represent an edge.
Specifically, we represent each edge $e \in E(T)$ by a copy of the unit interval $[0, 1]$ and connect these intervals according to their adjacencies.
We refer to the elements of $T$ as \emph{points} and to the elements of $V(T)$ as \emph{vertices}.
For ease of explanation, we identify $T$ with its topological realisation.

\begin{definition2}\label{def:merge-tree}
	A \emph{merge tree} is a pair $(T, f)$, where $T$ is a rooted tree and $f \colon \geom \to \R \cup \{\infty\}$ is a continuous height function that is strictly increasing towards the root, with $f(v) = \infty$ if and only if $v$ is the root.
\end{definition2}
A merge tree $(T, f)$ naturally induces a partial order on its points.
For two points $x_1, x_2 \in \geom$, the point $x_2$ is called an \emph{ancestor} of $x_1$, denoted $x_1 \preceq x_2$, if there is a path from $x_1$ to $x_2$ that is monotonically increasing under $f$.
We call $x_2$ a \emph{strict ancestor} of $x_1$ if $x_1 \prec x_2$.
Correspondingly, the point $x_1$ is called a \emph{(strict) descendant} of $x_2$ if $x_2$ is a (strict) ancestor of~$x_1$.
For a point $x \in \geom$, we denote by $T_x$ the subtree of $T$ rooted at~$x$.
Furthermore, for a given value $\delta \ge 0$, we denote by $x^\delta$ the unique ancestor of $x$ with $f(x^\delta) = f(x) + \delta$.

\begin{figure}[b]
	\centering
	\includegraphics{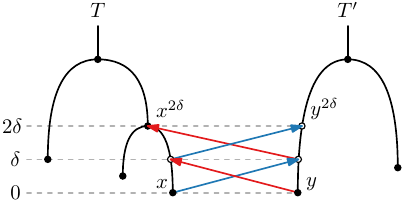}
	\caption{Two merge trees and part of a $\delta$-interleaving. Mapping a point $x$ from $\geom$ to $\geom'$ through~$\alpha$ (in blue), and mapping it back to $\geom$ via $\beta$ (in red) gives the unique ancestor $x^{2\delta}$ of $x$.}
	\label{fig:interleaving}
\end{figure}

\subparagraph{Interleaving distance.}
For two merge trees $(T, f)$ and $(T', f')$, we are now ready to define the \emph{interleaving distance}.
Let $\delta \ge 0$.
Intuitively, a $\delta$-interleaving describes a mapping $\alpha$ from $\geom$ to $\geom'$ that sends points exactly $\delta$ upwards, and a similar map $\beta$ from $\geom'$ to $\geom$, such that the compositions of $\alpha$ and $\beta$ send any point to its unique ancestor $2\delta$ higher.
Figure~\ref{fig:interleaving} shows an example of a $\delta$-interleaving.

\begin{definition2}[Morozov, Beketayev and Weber~\cite{morozov2013interleaving}]\label{def:interleaving}
	Given two merge trees $(T, f)$ and $(T', f')$, a pair of continuous maps~$\alpha \colon \geom \to \geom'$ and $\beta \colon \geom' \to \geom$ is called a \emph{$\delta$-interleaving} if for all $x \in \geom$ and $y \in \geom'$:
	\begin{multicols}{2}
	\begin{enumerate}[({C}1)]
		\item\label{C1}
			$f'(\alpha(x)) = f(x) + \delta$,
		
		\item\label{C2}
			$\beta(\alpha(x)) = x^{2\delta}$,
		
		\item\label{C3}
			$f(\beta(y)) = f'(y) + \delta$, and
			
		\item\label{C4}
			$\alpha(\beta(y)) = y^{2\delta}$.
	\end{enumerate}
	\end{multicols}
	\noindent The \emph{interleaving distance} $\intdist{}((T, f), (T', f'))$ is defined as the infimum of all $\delta$ such that there exists a $\delta$-interleaving.
\end{definition2}

\noindent
Gasparovich et al.~\cite{gasparovich2019intrinsic} showed that for finite trees there always exists an interleaving that achieves the infimum.
Hence, we can replace infimum with minimum.
The individual maps $\alpha$ and $\beta$ of a $\delta$-interleaving $(\alpha, \beta)$ are both \emph{$\delta$-shift maps}, i.e.\ continuous maps that send points exactly $\delta$ higher.
We say a $\delta$-interleaving is \emph{optimal} for two merge trees $T$ and $T'$ if their interleaving distance $\intdist{}(T, T')$ is $\delta$.

\subparagraph{Fréchet distance.}
The \emph{Fréchet distance} is a distance for curves.
We are specifically interested in the Fréchet distance for 1D curves.
A \emph{1D curve} is a continuous map $P \colon [0, 1] \to \R \cup \{\infty\}$.
The Fréchet distance between two 1D curves is defined in terms of matchings between the two curves.
A \emph{matching} between two curves $P$ and $Q$ is a continuous and increasing bijection $\mu \colon [0, 1] \to [0, 1]$.
We say a matching $\mu$ is a \emph{$\delta$-matching} if $|P(t) - Q(\mu(t))| \le \delta$ for all $t \in [0, 1]$.
The Fréchet distance is defined as the infimum value $\delta$ for which a $\delta$-matching exists:
\[
    \frechet(P, Q) = \inf_\mu \max_{t \in [0, 1]} |P(t) - Q(\mu(t))|.
\]

\subsection{Binary relations.}
In the remainder of the paper we use the following standard definitions from order theory.
A \emph{total order} $\le$ on a set $X$ is a binary relation that satisfies the following four properties.
\begin{itemize}
	\item \emph{Reflexivity.} For any $x \in X$, the relation $x \le x$ holds.
	\item \emph{Antisymmetry.} For any $x, y \in X$, if $x \le y$ and $y \le x$ then $x = y$.
	\item \emph{Transitivity.} For any $x, y, z \in X$, if $x \le y$ and $y \le z$ then $x \le z$.
    \item \emph{Totality.} For any $x, y \in X$, either $x \le y$ or $y \le x$ holds.
\end{itemize}
A total order $\le$ on a set $X$ induces a \emph{strict order} $<$ as follows.
For two points $x, y \in X$, we have $x < y$ if and only if $x \le y$ and not $y \le x$, or, equivalently, $x \le y$ and not $x = y$.

\section{An Order-Preserving Interleaving Distance}\label{sec:total-orders}
In this section, we introduce a form of \emph{ordered merge trees} and define an order-preserving distance on them.
We first introduce two classes of ordered merge trees and we show that they are equivalent.
In Section~\ref{subsec:monotone-interleaving-distance}, we then define the monotone interleaving distance.


\subsection{Ordered merge trees.}
We now define a new class of merge trees, which we call \emph{ordered merge trees}.
An ordered merge tree is a merge tree equipped with a total order on each of its level sets.
Specifically, consider a merge tree $(T, f)$ and assume that its lowest leaf has height exactly $0$.
If not, we can simply shift $f$ by a constant term.
Fix a height value $h \ge 0$.
The level set of $T$ at height $h$ is defined as $\ls{h} \coloneqq \{x \in \geom \mid f(x) = h\}$.
For a point $x \in T$ with $f(x) \le h$, we denote by $\an{x}{h}$ the unique ancestor of $x$ at height $h$.
We define a \emph{layer-order} as a set of total orders~$(\le_h)_{h \ge 0}$ on the level sets of $T$, such that it is \emph{consistent}: for two heights $h_1 \le h_2$ and two points $x_1, x_2$ in $\ls{h_1}$, we require that $x_1 \le_{h_1} x_2$ implies $\an{x_1}{h_2} \le_{h_2} \an{x_2}{h_2}$.

\begin{definition2}\label{def:ordered-merge-tree}
	An \emph{ordered merge tree} is a triple $(T, f, (\le_h)_{h\ge 0})$, with $(T, f)$ a merge tree and $(\le_h)_{h \ge 0}$ a layer-order.
\end{definition2}

\begin{figure}
	\centering
     \begin{subfigure}[b]{.49\textwidth}
        \centering
        \includegraphics{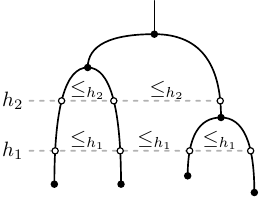}
        \subcaption{}
        \label{fig:ordered-merge-tree}
     \end{subfigure}
     \begin{subfigure}[b]{0.3\textwidth}
		\centering
		\includegraphics{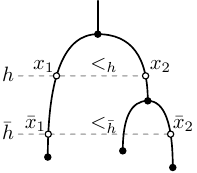}
		\subcaption{}
		\label{subfig:descendant-consistency}
	\end{subfigure}
	\caption{\enumfig{(a)} Visualisation of an ordered merge tree. \enumfig{(b)} Lemma~\protect\ref{lem:total-order-prop}: if $x_1 <_h x_2$, then any descendants $\bar{x_1} \preceq x_1$ and $\bar{x_2} \preceq x_2$ at the same height $\bar{h}$ also satisfy $\bar{x_1} <_{\bar{h}} \bar{x_2}$.}
\end{figure}

\noindent
We draw ordered merge trees such that the left-to-right order of the drawing respects the tree's layer-order (see Figure~\ref{fig:ordered-merge-tree}).
Moreover, we abbreviate the set~$(\le_h)_{h \ge 0}$ as $(\le_h)$.
For each value $h \ge 0$, the total order $\le_h$ induces a strict total order $<_h$ in the standard way.
Let $(T, f, (\le_h))$ be an ordered merge tree.
We first prove the following lemma, that extends the consistency property to subtrees rooted at distinct points (see Figure~\ref{subfig:descendant-consistency}).

\begin{figure}[b]
	\centering	
     \begin{subfigure}[b]{.49\textwidth}
        \centering
        \includegraphics{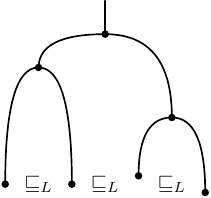}
        \subcaption{}
        \label{fig:leaf-ordered-merge-tree}
     \end{subfigure}
	\begin{subfigure}[b]{0.3\textwidth}
		\centering
		\includegraphics{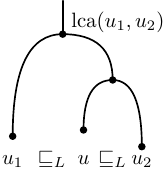}
		\subcaption{}
		\label{subfig:leaf-order}
	\end{subfigure}
	\caption{Visualisations of \enumfig{(a)} a leaf-ordered merge tree and \enumfig{(b)} the separating subtrees property of a leaf-order.}
	\label{fig:omt-examples}
\end{figure}

\begin{lemma}\label{lem:total-order-prop}
	Fix $h \ge 0$ and let $x_1, x_2 \in \ls{h}$ be two distinct points such that $x_1 <_h x_2$.
	Then also $\bar{x}_1 <_{\bar{h}} \bar{x}_2$ for any two points $\bar{x}_1 \in T_{x_1}$ and $\bar{x}_2 \in T_{x_2}$ at height $\bar{h} \le h$.
\end{lemma}
\begin{proof}
	Consider two points $\bar{x}_1 \in T_{x_1}$ and $\bar{x}_2 \in T_{x_2}$ at height $\bar{h} \le h$.
	Assume for the sake of a contradiction that not $\bar{x}_1 <_{\bar{h}} \bar{x}_2$.
	It follows that then $\bar{x}_2 <_{\bar{h}} \bar{x}_1$, as $\le_{\bar{h}}$ is a total order.
	We know that $x_1$ and $x_2$ are the unique ancestors at height $h$ of $\bar{x}_1$ and $\bar{x}_2$ respectively.
	However, then by consistency of $(\le_h)$ we must have $x_2 \le_h x_1$, contradicting our assumption.
\hfill\end{proof}

\noindent
If all leaves of an ordered merge tree are at the same height, say $h = 0$, then by the consistency property, the order $\le_0$ is a total order on the leaves and completely and uniquely defines the layer-order $(\le_h)$.
For trees that do not have their leaves all at the same height, we wish to similarly use the layer-order to infer a single total order on the leaves from which we can recover the layer-order.
For this purpose, we introduce a second class of ordered merge trees, which we call \emph{leaf-ordered merge trees} (see Figure~\ref{fig:leaf-ordered-merge-tree}).
Specifically, denote the \emph{lowest common ancestor} for a pair of points $x_1, x_2 \in \geom$ by $\lca(x_1, x_2) \in \geom$.
We define a \emph{leaf-order} as a total order $\leafo$ on $L(T)$ such that it \emph{separates} subtrees.
That is, for every leaf $u \in L(T)$, we require that if there exist leaves $u_1, u_2 \in L(T)$ with $u_1 \leafo u \leafo u_2$, then $u \in T_{\lca(u_1, u_2)}$ (see Figure~\ref{subfig:leaf-order}).

\begin{definition2}
    A \emph{leaf-ordered merge tree} is a triple $(T, f, \leafo)$ with $(T, f)$ a merge tree and~$\leafo$ a leaf-order.
\end{definition2}

\noindent
We show in Theorem~\ref{thm:bijection} that ordered merge trees and leaf-ordered merge trees are equivalent, by providing a bijection between the set of layer-orders and the set of leaf-orders.

\subparagraph{Induced leaf-ordered merge trees.}
Consider an ordered merge tree $(T, f, (\le_h))$.
First, we define the \emph{induced leaf-ordered merge tree} $\mathcal{L}((T, f, (\le_h))) \coloneqq (T, f, \ileaf)$, where $\ileaf$ is defined as follows:
for $u_1, u_2 \in L(T)$, we set $h = \max(f(u_1), f(u_2))$ and $u_1 \ileaf u_2$ if and only if $\an{u_1}{h} \le_h \an{u_2}{h}$.
We refer to $\ileaf$ as the \emph{induced leaf-order}.
It is fairly straightforward to show that $\ileaf$ is a leaf-order.

\begin{lemma}\label{lem:leaf-total}
    The induced leaf-order $\ileaf$ is a total order on the leaves of $T$.
\end{lemma}
\begin{proof}
	Reflexivity follows from the definition.
	To see that $\ileaf$ is antisymmetric, consider $\leaf_1, \leaf_2 \in L(T)$ and assume both $\leaf_1 \ileaf \leaf_2$ and $\leaf_2 \ileaf \leaf_1$.
	Let $h = \max(f(\leaf_1), f(\leaf_2))$.
	By definition, we have $\an{\leaf_1}{h} \le_h \an{\leaf_2}{h}$ and $\an{\leaf_2}{h} \le_h \an{\leaf_1}{h}$.
	Since $\le_h$ is a total order, we obtain $\an{\leaf_1}{h} = \an{\leaf_2}{h}$.
	One of $\leaf_1$, $\leaf_2$ has height $h$, say $f(\leaf_1) = h$.
	Consequentially, we get $\leaf_1 = \an{\leaf_1}{h} = \an{\leaf_2}{h}$, implying that $\leaf_1$ is an ancestor of $\leaf_2$.
	As $\leaf_1$ is a leaf, it does not have any strict descendants, so $\leaf_2 = \leaf_1$.
	A similar argument holds if $f(\leaf_2) = h$.
 
	Next, to show $\ileaf$ is transitive, take distinct points $\leaf_1, \leaf_2, \leaf_3 \in L(T)$ and assume $\leaf_1 \ileaf \leaf_2$ and $\leaf_2 \ileaf \leaf_3$.
	We need to show that $\leaf_1 \ileaf \leaf_3$.
	Define $h = \max(f(\leaf_1), f(\leaf_2), f(\leaf_3))$.
	By consistency of $(\le_h)$, we have $\an{\leaf_1}{h} \le_{h} \an{\leaf_2}{h}$ and $\an{\leaf_2}{h} \le_{h} \an{\leaf_3}{h}$.
	Using transitivity of $\le_h$, we also get $\an{\leaf_1}{h} \le_{h} \an{\leaf_3}{h}$.
	Now, set $h_1 = \max(f(\leaf_1), f(\leaf_3))$ and note $h_1 \le h$.
	If $h = h_1$, then $\leaf_1 \ileaf \leaf_3$ follows.
	Else, if $h_1 < h$, we must have $h = f(\leaf_2)$.
	Then $\an{\leaf_1}{h} \le_h \an{\leaf_2}{h} = \leaf_2$, and as the leaf $\leaf_2$ can not have descendants, we know that $\an{\leaf_1}{h} <_h \leaf_2$.
	Similarly, we know that $\leaf_2 <_h \an{\leaf_3}{h}$.
	By transitivity, we obtain $\an{\leaf_1}{h} <_h \an{\leaf_3}{h}$.
	Now, we can use Lemma~\ref{lem:total-order-prop} to get $\an{\leaf_1}{h_1} <_{h_1} \an{\leaf_3}{h_1}$, showing $\leaf_1 \ileaf \leaf_3$.
	Lastly, to show that $\ileaf$ is total, let $u_1, u_2 \in L(T)$.
    For any $h \ge 0$, the relation $\le_h$ is a total order.
    In particular, for $h = \max(f(u_1), f(u_2))$, either $\an{u_1}{h} \le \an{u_2}{h}$ or $\an{u_2}{h} \le \an{u_1}{h}$ holds, and hence $u_1 \irel u_2$ or $u_2 \irel u_1$.\hfill
\hfill\end{proof}

\begin{lemma}\label{lem:separating-subtrees}
	The induced leaf-order $\ileaf$ separates subtrees.
\end{lemma}
\begin{proof}
	Consider distinct leaves $\leaf, \leaf_1, \leaf_2 \in L(T)$ such that $\leaf_1 \ileaf \leaf \ileaf \leaf_2$.
	Define $v = \lca(\leaf_1, \leaf_2)$.
	Set $h = \max(f(v), f(\leaf))$, and observe that both $\an{\leaf_1}{h} = \an{v}{h}$ and $\an{\leaf_2}{h} = \an{v}{h}$.
	Since $\leaf_1 \ileaf \leaf$, we get $\an{v}{h} = \an{\leaf_1}{h} \le_h \an{\leaf}{h}$.
	Similarly, since $\leaf \ileaf \leaf_2$, we obtain $\an{\leaf}{h} \le_h \an{v}{h}$.
	Combining these inequalities, and using the antisymmetry of $\le_h$, we thus have $\an{\leaf}{h} = \an{v}{h}$.
	We now claim that $h = f(v)$, so that $\an{\leaf}{h} = v$, implying $\leaf \in T_v$.
	Assume for the sake of a contradiction that $h = f(\leaf)$.
	Then, we get $\an{v}{h} = \an{\leaf}{h} = \leaf$, implying that $\leaf_1 \prec v \preceq \leaf$.
	However, $\leaf$ is a leaf, and can not have any descendants.\hfill
\hfill\end{proof}

\noindent
Combining Lemmas \ref{lem:leaf-total} and \ref{lem:separating-subtrees}, we see that $\mathcal{L}$ induces a leaf-ordered merge tree.

\begin{corollary}\label{cor:induced-leaf}
    $\mathcal{L}(T, f, (\le_h))$ is a leaf-ordered merge tree.
\end{corollary}

\subparagraph{Induced ordered merge trees.}
Consider a leaf-ordered merge tree $(T, f, \leafo)$.
For two sets of leaves $U_1, U_2 \subset L(T)$, we write $U_1 \leafo U_2$ if for all $u_1 \in U_1$ and for all $u_2 \in U_2$, we have $u_1 \leafo u_2$.
We now define the \emph{induced ordered merge tree} $\mathcal{T}(T, f, \leafo) \coloneqq (T, f, (\le_h))$, where $\le_h$ is defined as follows: for height $h \ge 0$ and points $x_1, x_2 \in \ls{h}$, we set $x_1 \le_h x_2$ if and only if $x_1 = x_2$ or if~$L(T_{x_1}) \leafo L(T_{x_2})$.
We refer to $(\le_h)$ as the \emph{induced layer-order}.

\begin{restatable}{lemma}{inducedlayerorder}\label{lem:layer-total}
	For each $h \ge 0$, the relation $\le_h$ is a total order.
\end{restatable}
\begin{proof}
	Fix $h \ge 0$.
	Reflexivity follows immediately.
	To show $\le_h$ is antisymmetric, take $x_1, x_2 \in \ls{h}$ and assume $x_1 \le_h x_2$ and $x_2 \le_h x_1$.
	Assume for the sake of a contradiction $x_1 \neq x_2$.
	By construction, we then must have both $L(T_{x_1}) \lel L(T_{x_2})$ and $L(T_{x_2}) \lel L(T_{x_1})$.
	This means that for $\leaf_1 \in L(T_{x_1})$ and $\leaf_2 \in L(T_{x_2})$, we have $\leaf_1 \lel \leaf_2$ and $\leaf_2 \lel \leaf_1$, implying $\leaf_1 = \leaf_2$.
	However, then $x_1 = \an{\leaf_1}{f(x_1)} = \an{\leaf_2}{f(x_2)} = x_2$, contradicting our assumption $x_1 \neq x_2$.
	To show transitivity, take distinct points $x_1, x_2, x_3 \in \ls{h}$ and assume $x_1 \le_h x_2$ and $x_2 \le_h x_3$.
	We need to show that $x_1 \le_h x_3$, i.e.\ that $L(T_{x_1}) \lel L(T_{x_3})$.
	Let $\leaf_k \in T_{x_k}$ for $k = 1,2,3$.
	By assumption, we have $\leaf_1 \lel \leaf_2$ and $\leaf_2 \lel \leaf_3$, so by transitivity of $\lel$ we obtain $\leaf_1 \lel \leaf_3$.
	As these leaves were chosen arbitrarily, we obtain $L(T_{x_1}) \lel L(T_{x_3})$, implying that indeed $x_1 \le_h x_3$.
	
	Next, take two points $x_1, x_2 \in \ls{h}$ and assume not $x_1 \le_h x_2$.
	Then we know that $x_1 \neq x_2$ and not $L(T_{x_1}) \lel L(T_{x_2})$.
	Since $\lel$ is a total order, this means that there are leaves $\leaf_1 \in L(T_{x_1})$ and $\leaf_2 \in L(T_{x_2})$ such that $\leaf_2 \lel \leaf_1$.
	We now show that $L(T_{x_2}) \lel L(T_{x_1})$.
	Consider a leaf $\leaf_2' \in L(T_{x_2})$.
	Assume, towards a contradiction, that $\leaf_1 \lel \leaf_2'$.
	Then we have $\leaf_2 \lel \leaf_1 \lel \leaf_2'$.
	By the separating subtree property of $\ileaf$, we obtain $\leaf_1 \in T_{\lca(\leaf_2, \leaf_2')}$.
	As $\leaf_2, \leaf_2'$ are both leaves in $T_{x_2}$, it follows that $\lca(\leaf_2, \leaf_2') \preceq x_2$, and thus $\leaf_1 \in T_{x_2}$.
	However, as $x_1 \neq x_2$, the subtrees $T_{x_1}$ and $T_{x_2}$ are disjoint, and we reach a contradiction.
	Hence, $L(T_{x_2}) \lel \leaf_1$.
	Secondly, consider another leaf $\leaf_1' \in L(T_{x_2})$.
	We now show that $\leaf_2' \lel \leaf_1'$.
	Assume, towards a contradiction, that $\leaf_1' \lel \leaf_2' \lel \leaf_1$.
	Using a similar argument as before, we obtain $\leaf_2' \in T_{x_1}$, contradicting our assumption $x_1 \neq x_2$.
	Hence, we conclude that indeed $L(T_{x_2}) \lel L(T_{x_1})$, and thus $x_2 \le_h x_1$.
	This shows that $\le_h$ is a total order.
\hfill\end{proof}

\begin{lemma}\label{lem:layer-consistent}
    The induced layer-order $(\le_h)$ is consistent.
\end{lemma}
\begin{proof}
    We need to show that for $h_1 \ge 0$ and any two points $x_1, x_2 \in \ls{h_1}$ with $x_1 \le_{h_1} x_2$ we also have $\an{x_1}{h_2} \le_{h_2} \an{x_2}{h_2}$ for all $h_2 \ge h_1$. 
	If we assume that $x_1 \le_{h_1} x_2$, then by construction we know that $L(T_{x_1}) \lel L(T_{x_2})$.
	Let $h_2 \ge h_1$ and set $\hat{x}_1 = \an{x_1}{h_2}$ and $\hat{x}_2 = \an{x_2}{h_2}$.
	If $x_1 = x_2$, it immediately follows that $\hat{x}_1 = \hat{x}_2$.
	Else, towards a contradiction, assume that we do not have $\hat{x}_1 \le_{h_2} \hat{x}_2$.
	We have already shown that $\le_{h_2}$ is a total order, so we must have $\hat{x}_2 <_{h_2} \hat{x}_1$.
	By construction, this means that $L(T_{\hat{x}_2}) \lel L(T_{\hat{x}_1})$.
	Observe that because $x_1 \preceq \hat{x}_1$, we have $L(T_{x_1}) \subseteq L(T_{\hat{x}_1})$.
	Similarly, we have $L(T_{x_2}) \subseteq L(T_{\hat{x}_2})$.
	It follows that $L(T_{x_2}) \lel L(T_{x_1})$
	Combining with our initial assumption, we get $L(T_{x_1}) = L(T_{x_2})$.
	Since $x_1$ and $x_2$ are from the same layer, this implies that $x_1 = x_2$.
	However, this contradicts our assumption that $x_1$ and $x_2$ are distinct points.
\hfill\end{proof}

\noindent
Combining Lemmas \ref{lem:layer-total} and \ref{lem:layer-consistent}, we see that $\mathcal{T}$ induces an ordered merge tree.

\begin{corollary}\label{cor:induced-ordered}
    $\mathcal{T}(T, f, \leafo)$ is an ordered merge tree.
\end{corollary}

\noindent
We now show that if we compose the maps $\mathcal{L}$ and $\mathcal{T}$, we obtain the identity maps between ordered merge trees and leaf-ordered merge trees.

\begin{restatable}{lemma}{identity1}\label{lem:identity-1}
    $\mathcal{T}(\mathcal{L}((T, f, (\le_h)))) = (T, f, (\le_h))$.    
\end{restatable}
\begin{proof}
    By Corollaries~\ref{cor:induced-leaf} and $\ref{cor:induced-ordered}$ we know that $\mathcal{L}((T, f, (\le_h)))$ is a leaf-ordered merge tree $(T', f', \ileaf)$, and that $\mathcal{T}(\mathcal{L}((T, f, (\le_h))))$ is an ordered merge tree $(T'', f'', (\le'_h))$.
    From the definitions of $\mathcal{L}$ and $\mathcal{T}$ it immediately follows that $T = T' = T''$ and $f = f' = f''$, so it suffices to show that $(\le_h)$ equals $(\le'_h)$.
    Fix $h \ge 0$ and let $x_1, x_2 \in \ls{h}$.
    Assume that $x_1 \le_h x_2$.
    If $x_1 = x_2$, then it immediately follows that $x_1 \le'_h x_2$, so assume $x_1 <_h x_2$.
    We need to show that $x_1 <'_h x_2$ holds as well.
    Let $u_1 \in L(T_{x_1})$ and $u_2 \in L(T_{x_2})$ be any two leaves in the subtrees rooted at $x_1$ and $x_2$.
    Set $h^{u} = \max(f(u_1), f(u_2))$.
    Since both $f(u_1) \le h$ and $f(u_2) \le h$, it follows that $h^u \le h$, and we can use Lemma~\ref{lem:total-order-prop} to obtain $\an{u_1}{h^u} <_{h^u} \an{u_2}{h^u}$.
    By definition of the induced leaf-order, it follows that $u_1 \ileaf u_2$.
    Since this holds for any two leaves $u_1$ and $u_2$, we know that $L(T_{x_1}) \ileaf L(T_{x_2})$.
    Hence, by definition of the induced layer-order, we know that indeed $x_1 <'_h x_2$.    
\hfill\end{proof}

\begin{restatable}{lemma}{identity2}\label{lem:identity-2}
    $\mathcal{L}(\mathcal{T}((T, f, \leafo))) = (T, f, \leafo)$.    
\end{restatable}
\begin{proof}
    By Corollaries~\ref{cor:induced-ordered} and $\ref{cor:induced-leaf}$ we know that $\mathcal{T}((T, f, \leafo)) = (T', f', (\le_h))$ is an ordered merge tree, and that $\mathcal{L}(\mathcal{T}((T, f, \leafo))) = (T'', f'', \ileaf')$ is a leaf-ordered merge tree.
    It immediately follows that $T = T' = T''$ and $f = f' = f''$, so it suffices to show that $\leafo$ equals $\ileaf'$.
    Take $u_1, u_2 \in L(T)$ and assume that $u_1 \leafo u_2$.
    Set $h = \max(f(u_1), f(u_2))$, and set $x_1 = \an{u_1}{h}$ and $x_2 = \an{u_2}{h}$.
    We need to show that $u_1 \ileaf' u_2$ holds as well, i.e. that $x_1 \le_h x_2$.
    By definition of the induced layer-order, this holds if and only if $x_1 = x_2$, or if $L(T_{x_1}) \leafo L(T_{x_2})$.
    If $u_1 = u_2$, we trivially have $x_1 = x_2$, so assume $u_1 \neq u_2$.
    Since both $u_1$ and $u_2$ are leaves and do not have any descendants, we can not have $x_1 = x_2$.
    Instead, we show that $L(T_{x_1}) \leafo L(T_{x_2})$.

    We consider two cases.
    First, if $h = f(u_1)$, then $x_1 = u_1$ and $L(T_{x_1}) = \{u_1\}$.
    Take $u \in L(T_{x_2})$, and suppose $u \ileaf u_1$.
    Then we have $u \ileaf u_1 \ileaf u_2$, and since $\ileaf$ is a leaf-order, we must have that $u_1 \in T_{\lca(u, u_2)}$.
    Since both $u$ and $u_2$ are leaves in $T_{x_2}$, we must have that $\lca(u, u_2) \in T_{x_2}$, and implying that $u_1 \in T_{x_2}$.
    It then follows that $u_1 = x_2$, as $f(u_1) = h = f(x_2)$.
    However, we assumed $u_1 \neq u_2$, and because $u_1$ is a leaf, $u_1$ can not have any descendants.
    We thus reach a contradiction, showing that $u_1 \leafo u$.
    Otherwise, if $h = f(u_2)$, we can use a similar argument showing that any for $u \in L(T_{x_1})$ we must have $u \leafo u_2$.
    We conclude that $L(T_{x_1}) \leafo L(T_{x_2})$.
\hfill\end{proof}

\noindent
Combining Lemmas~\ref{lem:identity-1} and \ref{lem:identity-2}, we obtain Theorem~\ref{thm:bijection}.
\begin{theorem}\label{thm:bijection}
    The maps $\mathcal{L}$ and $\mathcal{T}$ define a bijection between the set of ordered merge trees and the set of leaf-ordered merge trees.
\end{theorem}

\noindent
In the remainder of this paper, we interchange leaf-ordered and ordered merge tree freely, and, if clear from context, we refer to both of them simply as ordered merge trees.

\subsection{Monotone interleaving distance.}\label{subsec:monotone-interleaving-distance}
We are now ready to define an order-preserving distance on ordered merge trees.
Specifically, consider two ordered merge trees $(T, f, (\le_h))$ and $(T', f', (\le'_h))$.
We say a $\delta$-shift map $\alpha \colon \geom \to \geom'$ is \emph{layer-monotone}, or simply \emph{monotone}, if for all height values $h \ge 0$, and for any two points $x_1, x_2 \in \ls{h}$ it holds that if $x_1 \le_h x_2$ then also $\alpha(x_1) \le'_{h+\delta} \alpha(x_2)$.
A \emph{monotone $\delta$-interleaving}, in turn, is a $\delta$-interleaving $(\alpha, \beta)$ such that the maps $\alpha$ and $\beta$ are both monotone.

\begin{figure}[b]
	\centering
	\includegraphics{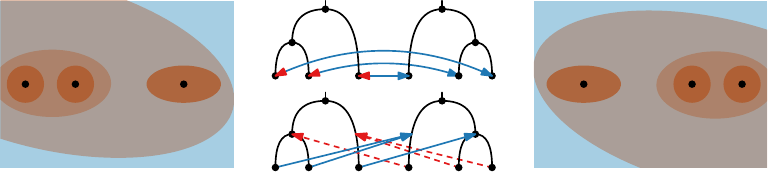}
	\caption{Two schematic rivers with three bars. The regular interleaving distance (top) is $0$, since the rivers induce the same merge trees. The monotone interleaving distance (bottom) is greater than $0$ and captures the left-to-right order of the bars better.}
	\label{fig:monotone-maps}
\end{figure}

\begin{definition2}
    The \emph{monotone interleaving distance} $\orddist{}$ is the minimum~$\delta$ that admits a monotone $\delta$-interleaving.
\end{definition2}
As a monotone interleaving is still an interleaving, the monotone interleaving distance is bounded from below by the regular interleaving distance.
On the contrary, Figure~\ref{fig:monotone-maps} illustrates that an optimal interleaving is not necessarily an optimal monotone interleaving.

\section{Relation to 1D Curves}\label{sec:frechet}
In this section, we establish a connection between the monotone interleaving distance and the Fréchet distance.
An in-order tree walk on an ordered merge tree induces a particular (in-order) curve on the tree, which in turn induces a particular 1D curve.
We show that the monotone interleaving distance between two ordered merge trees equals the Fréchet distance between their induced 1D curves.

\subparagraph{Curves.}
Let $(T, f, (\le_h))$ be an ordered merge tree.
We define a 1D curve $P$ by tracing $f$ along an \emph{in-order tree walk} of $T$.
For technical reasons, we require that such a walk starts and ends at the root (at infinity) of $T$.
We use $\mathrm{deg}(x)$ to refer to the (down)degree of a point $x \in T$.
Here, $\mathrm{deg}(x) = 1$ for points interior of edges of $T$, and $\mathrm{deg}(x) = 0$ for leaves of $T$.
For a curve $\sigma\colon [0, 1] \to T$ on a tree $T$ and a point $x\in T$, we say that the number of times that $\sigma$ \emph{visits} $x$ is the number of connected components of $\sigma^{-1}(x)$. If $\sigma^{-1}(x)$ is empty for some $x \in T$, then we say that $x$ is unvisited by $\sigma$.
For a vertex $v \in T$ and a strict ancestor\footnote{Possibly also a vertex $u \neq v$.} $x$ in the parent edge of $v$, we define the \emph{planted subtree} $T_{x, v} = T_v \cup [x, v]$.
We call $x$ the root of $T_{x, v}$.
We say that a planted subtree $T_{x, v}$ is \emph{$\sigma$-unvisited} if all points, except possibly the root, are unvisited by $\sigma$.
Lastly, we define the \emph{$\sigma$-unvisited degree} of any $x \in T$ as the number of $\sigma$-unvisited planted subtrees rooted at $x$.

\begin{definition2}
    Let $(T, f, (\le_h))$ be an ordered merge tree and let $\tau\colon [0,1]\to T$ be a curve that starts and ends at the root of $T$.
    We distinguish three properties for such a curve:
    \begin{enumerate}
        \item It respects the order of $T$, i.e. for all $t_1, t_2 \in [0, 1]$, $\tau(t_1) \le_h \tau(t_2)$ implies $t_1 \le t_2$.
        \item Each point $x \in T$ is visited exactly $\mathrm{deg}(x) + 1 - \kappa$ times, where $\kappa$ is the $\tau$-unvisited degree of $x$.
        \item Each point $x \in T$ is visited exactly $\mathrm{deg}(x) + 1$ times.
    \end{enumerate}
    We call a curve that satisfies
    \begin{description}
        \item \textit{1}: a \emph{weak in-order curve} on $T$,
        \item \textit{1-2} a \emph{partial in-order curve} on $T$,
        \item \textit{1-3:} an \emph{in-order curve} on $T$.
    \end{description}
\end{definition2}

\begin{figure}[b]
    \centering
    \includegraphics{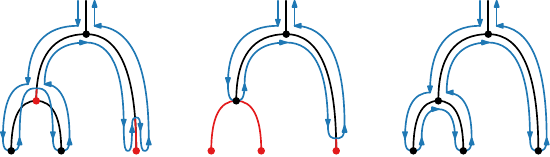}
    \caption{From left to right: a weak in-order curve, a partial in-order curve, and an in-order curve. The red points in the left tree violate property \enumit{2}; the red points in the middle tree violate property \enumit{3}.}
    \label{fig:in-order-curves}
\end{figure}
\noindent
See Figure~\ref{fig:in-order-curves} for examples of weak, partial, and regular in-order curves.

\begin{figure}[b]
    \centering
    \includegraphics{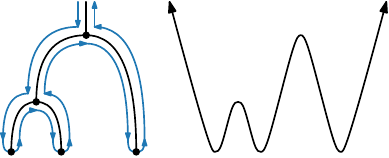}
    \caption{An in-order curve (in blue) on an ordered merge tree, and its corresponding induced 1D curve visualised in 2D; the $x$-coordinate indicates the parameter value.}
    \label{fig:induced-curve}
\end{figure}

\noindent
Given an in-order curve $\tau$ on $T$, we define the 1D curve $P_\tau = f \circ \tau$; that is, by tracing the height of $T$ along $\tau$.
Note that $P_\tau(0) = \infty = P_\tau(1)$.
See Figure~\ref{fig:induced-curve} for an example of an in-order curve and its corresponding 1D curve.
The 1D curve $P_\tau$ is unique up to reparameterisation of $\tau$, so we usually omit $\tau$ and refer to $P$ as the \emph{induced 1D curve} of the ordered merge tree $(T, f, (\le_h))$.
A (partial) in-order curve has the following useful property:

\begin{lemma}\label{lem:subcurves}
    Let $\tau \colon [0, 1] \to T$ be a (partial) in-order curve and take $t_1, t_2 \in [0, 1]$.
    The subcurve from $t_1$ to $t_2$ visits only points in the subtree $T_x$, where $x = \lca(\tau(t_1), \tau(t_2))$.
\end{lemma}
\begin{proof}
    Let $e$ be the parent edge of $x$.
    Because $\tau$ starts and ends at the root, and $\tau(t_1)$ and $\tau(t_2)$ are descendants of $x$, the edge $e$ must be traversed at least once on the subcurve $[0,t_1]$, and once on the subcurve $[t_2,1]$.
    Points $x'$ interior of $e$ have $\mathrm{deg}(x')=1$, so if $\tau$ is partial in-order, no point interior of $e$ can be visited by the subcurve $[t_1,t_2]$.
    Hence, the subcurve $[t_1,t_2]$ visits only points in the subtree $T_x$.
\hfill\end{proof}

\noindent
Note that if $x \coloneqq \tau(t_1) = \tau(t_2)$, then $\lca(\tau(t_1), \tau(t_2)) = x$.
This gives the following corollary.
\begin{corollary}\label{cor:subcurves}
    Let $\tau \colon [0, 1] \to T$ be a (partial) in-order curve.
    For $x \in T$ and $t_1, t_2 \in \tau^{-1}(x)$, the subcurve $[t_1, t_2]$ visits only points in the subtree $T_x$.
\end{corollary}

\subparagraph{Matched curves.}
Let $(T, f, (\le_h))$ and $(T', f', (\le'_h))$ be two ordered merge trees, and let $\tau \colon [0, 1] \to T$ and $\tau' \colon [0, 1] \to T'$ be two in-order curves on $T$ and $T'$ respectively.
The Fréchet distance between the induced 1D curves $P_\tau$ and $Q_{\tau'}$ is the infimum value $\delta$ for which a matching $\mu$ exists such that $|f(\tau(t)) - f'(\tau'(\mu(t)))| \le \delta$ for all $t \in [0, 1]$.
The composition $\tau' \circ \mu$ is a reparameterisation of $\tau'$, and therefore also an in-order curve of $T'$.
Instead of considering matchings $\mu$, it suffices to consider in-order curves $\tau$ and $\tau'$ on $T$ and $T'$.
A pair of in-order curves $(\tau, \tau')$ is \emph{$\delta$-matched} if $|f(\tau(t)) - f'(\tau'(t))| \le \delta$ for all $t \in [0, 1]$.
The Fréchet distance between induced 1D curves $P$ and $Q$ can equivalently be defined as the infimum value $\delta$ for which a pair of $\delta$-matched in-order curves on $T$ and $T'$ exists: 
\[
    \frechet(P, Q) = \inf_{(\tau, \tau')} \max_{t \in [0, 1]} |f(\tau(t)) - f'(\tau'(t))|.
\]
\noindent
Our main result is the following:

\begin{restatable}{theorem}{frechetequals}\label{thm:frechet}
	$\frechet(P, Q) = \orddist{}(T, T')$.
\end{restatable}

\noindent
To prove Theorem~\ref{thm:frechet}, we describe constructions from $\delta$-matched in-order curves to monotone $\delta$-interleavings, and vice versa.

\subsection{From matched in-order curves to monotone interleavings.}
Let $\tau \colon [0, 1] \to T$ and $\tau' \colon [0, 1] \to T'$ be a $\delta$-matched pair of in-order curves on $T$ and $T'$ respectively.
Then $|f(\tau(t)) - f'(\tau'(t))| \le \delta$ for all $t \in [0, 1]$.
We define a map $\alpha \colon T \to T'$ as follows.
For $x \in T$, fix $t \in \tau^{-1}(x)$ arbitrarily and set $\alpha(x) = \an{\tau'(t)}{f(x)+\delta}$.
To see that such an ancestor exists, it suffices to show that $f'(\tau'(t)) \le f(x) + \delta$, which is evident because $\tau$ and $\tau'$ are $\delta$-matched.
We now first show that the definition of $\alpha(x)$ is independent from the choice of $t \in \tau^{-1}(x)$.

\begin{lemma}\label{lem:alpha-well-defined}
   Fix $x \in T$ and let $h = f(x)$.
   For all $t_1, t_2 \in \tau^{-1}(x)$, we have $\an{\tau'(t_1)}{h+\delta} = \an{\tau'(t_2)}{h+\delta}$.
\end{lemma}
\begin{proof}
    Define $y \coloneqq \an{\tau'(t_1)}{h+\delta}$.
    By Lemma~\ref{lem:subcurves} we know that the subcurve $[t_1, t_2]$ lies in $T_x$.
    So, for all $t \in [t_1, t_2]$ we have $f(\tau(t)) \le h$, and thus also $f'(\tau'(t)) \le h + \delta = f'(y)$.
    By continuity of $\tau'$, it follows that $\tau'(t) \preceq y$ for all $t \in [t_1, t_2]$, and thus that $\an{\tau'(t_2)}{h+\delta} = y$.
\hfill\end{proof}

\noindent
This shows that the map $\alpha$ is well-defined, based on $\tau$ and $\tau'$.
We now show that $\alpha$ is a monotone $\delta$-shift map.
\begin{lemma}
    The map $\alpha$ is a monotone $\delta$-shift map.    
\end{lemma}
\begin{proof}
    We first argue that $\alpha$ is continuous.
    For this, we show that $\alpha$ is continuous at every point $x \in T$.
    Fix $x \in T$ and set $h = f(x)$.
    Since all edges of a merge tree have positive length, we can consider a small neighbourhood $N$ of $x$, in which all points are reachable from $x$ by an $f$-monotone path.
    We define the distance between $x$ and a point $x' \in N$ as the length of such an $f$-monotone path.
    To show continuity of $\alpha$ at $x$, we show that for any $x' \in N$, the point $\alpha(x)$ is an ancestor of $\alpha(x')$, or vice versa.
    This implies the existence of an $f'$-monotone path in $T'$ between $\alpha(x)$ and $\alpha(x')$, whose length (induced by $f'$) is $|f'(\alpha(x))-f'(\alpha(x'))|=|f(x)-f(x')|$ and hence equal to the distance between $x$ and $x'$.
    
    A point $x'$ in a small neighbourhood of $x$ either has $f(x')<f(x)$ or $f(x')>f(x)$.
    If $f(x') < f(x)$, then $x'$ is a descendant of $x$.
    Let $t_1, t_2 \in \tau^{-1}(x)$ be the first and last time $x$ is visited by $\tau$.
    As $\tau$ is an in-order curve on $T$, we know that $x'$ is visited by the subcurve $[t_1, t_2]$.
    Let $t' \in \tau^{-1}(x')$, so that $t_1 \le t' \le t_2$, and set $y = \lca(\tau'(t_1), \tau'(t_2))$.
    By Lemma~\ref{lem:subcurves}, it then follows that the point $\tau'(t)$ lies in the subtree $T'_y$.
    Moreover, by Lemma~\ref{lem:alpha-well-defined}, we know that $\an{\tau'(t_1)}{h+\delta} = \an{\tau'(t_2)}{h+\delta}$.
    As a result, we know that $y$ is an descendant of $\an{\tau'(t_1)}{h+\delta} = \alpha(x)$, and thus that $\tau'(t) \preceq \alpha(x)$.
    Hence, as we assumed $f(x') < f(x)$, it follows that $\alpha(x') \preceq \alpha(x)$.
    If $f(x') > f(x)$, then $x$ is a descendant of $x'$.
    We can use the same arguments as before to reason that $\alpha(x) \preceq \alpha(x')$.
    This shows that $\alpha$ is indeed continuous.
    
    
    To see that $\alpha$ is monotone, let $x_1, x_2 \in T$ such that $x_1 \le_h x_2$, where $h = f(x_1) = f(x_2)$.
    Since $\tau$ is an in-order curve, we know that for any $t_1 \in \tau^{-1}(x_1)$ and $t_2 \in \tau^{-1}(x_2)$ we have $t_1 \le t_2$.
    Towards a contradiction, assume $\alpha(x_2) <'_h \alpha(x_1)$.
    Then $\an{\tau'(t_2)}{f(x)+\delta} <'_h \an{\tau'(t_1)}{f(x)+\delta}$, and by Lemma~\ref{lem:total-order-prop} it follows that $\tau'(t_2) <'_h \tau'(t_1)$.
    However, then it must be that $t_2 < t_1$, and we reach a contradiction.
    Lastly, by construction $\alpha$ maps points exactly $\delta$ higher, showing that it is a $\delta$-shift map.
\hfill\end{proof}

\noindent
We define a map $\beta \colon T' \to T$ symmetrically.
By exchanging the roles of $T$ and $T'$, and $\tau$ and $\tau'$ an analogous proof shows that $\beta$ is a monotone $\delta$-shift map.

\begin{lemma}
    The pair $(\alpha, \beta)$ is a monotone $\delta$-interleaving.
\end{lemma}
\begin{proof}
We need to show that for all $x \in T$ it holds that $\beta(\alpha(x)) = x^{2\delta}$.
Fix $x \in T$, set $h = f(x)$, and take $t \in \tau^{-1}(x)$.
By construction, we have $\alpha(x) = \an{\tau'(t)}{h+\delta}$.
Denote this point $y \coloneq \alpha(x)$.
We distinguish two cases.
If $y = \tau'(t)$, then since $t \in \tau'^{-1}(y)$ we obtain $\beta(\alpha(x)) = \beta(y) = \an{\tau(t)}{h+2\delta} = x^{2\delta}$.
Otherwise, if $y \neq \tau'(t)$, then $y$ must be a strict ancestor of $\tau'(t)$.
Therefore, we know that $\tau$ visits $y$ at least once on the path from root to $\tau'(t)$, and at least once on the path from $\tau'(t)$ to the root.
In other words, there are $t_1, t_2 \in \tau'^{-1}(y)$ such that $t_1 < t < t_2$.

By construction, we have that $\beta(y) = \an{\tau(t_1)}{h+2\delta} = \an{\tau(t_2)}{h+2\delta}$.
We see that $f(\beta(y)) = h + 2\delta$, so it suffices to show that $x \in T_{\beta(y)}$.
We again distinguish two cases.
First, consider the case that both $\tau(t_1) = \beta(y)$ and $\tau(t_2) = \beta(y)$.
Since $t_1 < t < t_2$ and $x = \tau(t)$, we can use Corollary~\ref{cor:subcurves} to argue that $x \in T_{\beta(y)}$.
Otherwise, if not both $\tau(t_1) = \beta(y)$ and $\tau(t_2) = \beta(y)$, it must be the case that $\tau$ visits $\beta(y)$ on the path from the root to $\tau(t_1)$, and on the path from $\tau(t_2)$ to the root.
This implies that there are $t_3, t_4 \in \tau^{-1}(\beta(y))$ such that $t_3 < t_1$ and $t_2 < t_4$.
So, we have $t_3 < t < t_4$.
Again, it follows by Corollary~\ref{cor:subcurves} that the point $\tau(t)$ lies in the subtree $T_{\beta(y)}$.
A symmetric argument shows that $\alpha(\beta(y)) = y^{2\delta}$ for all $y \in T'$.
\hfill\end{proof}

\noindent
It follows that we can construct a monotone $\delta$-interleaving from any pair of $\delta$-matched in-order curves.
\begin{corollary}\label{cor:dmi_leq_df}
    $\orddist{}(T, T') \le \frechet(f\circ\tau, f'\circ\tau')$.
\end{corollary}

\subsection{From monotone interleavings to matched in-order curves.}
Let $(\alpha, \beta)$ be a monotone $\delta$-interleaving between $T$ and $T'$.
Let $\tau$ be any in-order curve on $T$.
We use $\tau$ to construct an in-order curve $\tau'$ on $T'$.
First, initialise a map $\tau'_0 \colon [0, 1] \to T'$ by taking the image of $\tau$ on $T'$ under $\alpha$.
That is, we set $\tau'_0(t) = \alpha(\tau(t))$.
Because $\alpha$ is continuous, $\tau'_0$ is a curve, and moreover because $\alpha$ is monotone, $\tau'_0$ is a weak in-order curve on $T'$.
Furthermore, by construction, $\tau'_0$ satisfies $f'(\tau'_0(t)) - f(\tau(t)) = \delta$ for all $t \in [0, 1]$.

The map $\alpha$ is not necessarily injective or surjective on $T'$, so $\tau'_0$ might visit points $y \in T'$ more than $\mathrm{deg}(y) + 1$ times, and might not visit some points at all.
Hence, $\tau'_0$ is not necessarily a (partial) in-order curve.
We first transform $\tau'_0$ into a partial in-order curve $\tau'_1$ on $T'$ by contracting subcurves that are visited too often, and then transform $\tau'_1$ into an in-order curve $\tau'$ by extending the curve to visit subtrees that are not visited. For the intermediate curves of these transformations, we maintain that (under $f'$) they lie within distance $\delta$ of $f\circ\tau$.

\begin{figure}[b]
    \centering
    \includegraphics{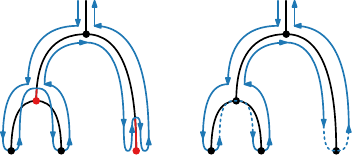}
    \caption{We resolve violating subcurves by moving them down.}
    \label{fig:violating-subcurve}
\end{figure}

\subparagraph{Violating subcurves.}
Let $\sigma \colon [0, 1] \to T$ be an arbitrary curve on an ordered merge tree $(T, f, \le_h)$.
Call a subcurve $[\ell, r]$ of $\sigma$ \emph{violating} if $\ell < r$, $\sigma(\ell) = \sigma(r)$, and for all $t \in (\ell, r)$ it holds that $f(\sigma(t)) > f(\sigma(\ell))$ (see Figure~\ref{fig:violating-subcurve}).
A violating subcurve is \emph{maximal} if it is not contained in any other violating subcurve of $\sigma$.
We show that two maximal violating subcurves cannot overlap.

\newpage
\begin{lemma}
    Any two maximal violating subcurves of a curve $\sigma$ on $T$ are interior-disjoint.
\end{lemma}
\begin{proof}
    Suppose for a contradiction that two maximal violating subcurves $[\ell_1, r_1]$ and $[\ell_2, r_2]$ of $\sigma$ are not interior-disjoint, and assume without loss of generality that $\ell_1 \le \ell_2$.
    By maximality, we have $\ell_1 < \ell_2 < r_1 < r_2$.
    By definition of violating subcurves, we have $f(\sigma(\ell_2)) > f(\sigma(\ell_1)) = f(\sigma(r_1))$ and simultaneously $f(\sigma(r_1)) > f(\sigma(\ell_2))$, reaching a contradiction.
\hfill\end{proof}

\noindent
We will show in Lemma~\ref{lem:partialVsViolating} that a weak in-order curve $\sigma$ on $T$ is a partial in-order curve if and only if it does not have any violating subcurve.
For this, we use the following property.
\begin{lemma}\label{lem:weakAbovefx}
    For any weak in-order curve $\sigma \colon [0, 1] \to T$, if $\sigma(t_1)=\sigma(t_2)=x$, then for all $t\in[t_1,t_2]$ with $f(\sigma(t))=f(x)$, we have $\sigma(t)=x$.
\end{lemma}
\begin{proof}
    If instead $\sigma(t)=x'\neq x$ for some $t$ with $f(\sigma(t))=f(x)$, then because $\sigma$ is weak in-order, we would have $x<_{f(x)} x'<_{f(x)} x$, which is impossible.
\hfill\end{proof}

\begin{lemma}\label{lem:partialVsViolating}
    A weak in-order curve $\sigma \colon [0, 1] \to T$ is partial if and only if it does not have any violating subcurve.
\end{lemma}
\begin{proof}
    If $\sigma$ has a violating subcurve, it does not satisfy Lemma~\ref{lem:subcurves}.
    Hence, it cannot be a partial in-order curve.

    \noindent
    To show the converse, suppose that $\sigma$ does not have violating subcurves.
    We first show that (for any $x\in T$) $\sigma$ must remain in the subtree of $x$ between consecutive visits of $x$.
    Let $t_1<t_2$ and assume that $\sigma(t_1)=\sigma(t_2)=x$, and $\sigma(t)\neq x$ for any $t\in(t_1,t_2)$.
    We show that the subcurve $[t_1,t_2]$ of $\sigma$ lies in $T_x$.
    If not, then there exists some $t\in(t_1,t_2)$ such that $f(\sigma(t))>f(x)$.
    Because the subcurve $(t_1,t_2)$ does not visit $x$, it follows from Lemma~\ref{lem:weakAbovefx} and continuity of $\sigma$ that $f(\sigma(t'))>f(x)$ for all $t'\in(t_1,t_2)$.
    Hence, $[t_1,t_2]$ is a violating subcurve, and we reach a contradiction.
    So indeed, the subcurve $[t_1, t_2]$ of $\sigma$ lies in $T_x$.
    
    We can now show that $\sigma$ visits each point $x \in T$ at most $\mathrm{deg}(x) + 1 - \kappa$ times, where $\kappa$ is the unvisited degree of $x$.
    Suppose for a contradiction that $\sigma$ visits $x$ at least $\mathrm{deg}(x) + 2 - \kappa$ times.
    Then by the pigeon hole principle, there are at least two subpaths of $\sigma$ that start and end at $x$, and both visit some strict descendant $x'$ of $x$.
    Hence, there exists a subpath of $\sigma$ that starts and ends at $x'$ and visits $x\notin T_{x'}$, which contradicts that any such subpath must remain in $T_{x'}$.
\hfill\end{proof}

\noindent 
We modify $\tau'_0$ by removing all violating subcurves.
Specifically, let $\mathcal{S} = \{[\ell_1, r_1], \ldots, [\ell_k, r_k]\}$ be the set of maximal violating subcurves of $\tau'_0$, and define
\[
    \tau'_1(t) = \begin{cases}
        \tau'_0(\ell_i) & \text{for } t \in [\ell_i, r_i] \in \mathcal{S},\\
        \tau'_0(t) & \text{otherwise.}
    \end{cases}
\]

\noindent
We argue that the resulting curve $\tau'_1$ is a partial in-order curve on $T'$.
\begin{lemma}
    The curve $\tau'_1$ is a partial in-order curve on $T'$.
\end{lemma}
\begin{proof}
    If $\tau'_1$ has a violating subcurve $[t_1,t_2]$, then that subcurve is also violating in $\tau'_0$, so $[t_1,t_2]\subseteq [\ell_i,r_i]$ for some $i$. But then the subcurve $[t_1,t_2]$ of $\tau'_1$ is constant, and hence not violating.
    So $\tau'_1$ does not have any violating subcurve, so by Lemma~\ref{lem:partialVsViolating} it is a partial in-order curve on $T'$.
\hfill\end{proof}

\newpage
\noindent
The curve $\tau'_1$ is constant on each interval in $\mathcal{S}$.
We refer to a maximal interval\footnote{Possibly a singleton interval.} on which $\tau'_1$ is constant as a \emph{paused interval}.
We argue that $\tau'_1$ has the following two properties:
\begin{enumerate}[(i)]
    \item for each paused interval $[t_1, t_2]$ of $\tau'_1$, it holds that $f'(\tau'_1(t_1)) - f(\tau(t_1)) = \delta$, and
    \item all values $t \in [0,1]$ satisfy $|f(\tau(t)) - f'(\tau'_1(t))| \le \delta$.
\end{enumerate}

\begin{lemma}\label{lem:property-1}
    The curve $\tau'_1$ satisfies property \enumit{(i)}.
\end{lemma}
\begin{proof}
    For all $t\in[0,1]$, we have $f'(\tau'_0(t)) - f(\tau(t)) = \delta$.
    If $[t_1,t_2]$ is a paused interval of~$\tau'_1$, then $t_1$ did not lie interior of any violating subcurve of $\tau'_0$, and hence $\tau'_1(t_1)=\tau'_0(t_1)$, completing the proof.
\hfill\end{proof}

\noindent
To prove the second property, we bound the height of violating subcurves.
\begin{lemma}\label{lem:property-2}
    The curve $\tau'_1$ satisfies property \enumit{(ii)}.
\end{lemma}
\begin{proof}
    Fix $t \in [0, 1]$.
    We have $f'(\tau'_0(t)) - f(\tau(t)) = \delta$.
    If $t$ is not contained in any violating subcurve of $\tau'_0$, then $\tau'_1(t) = \tau'_0(t)$ and we are done.
    Otherwise, $t$ is contained in some violating subcurve $[\ell,r]$ of $\tau'_0$, and we must show that $f'(\tau'_0(t))-f'(\tau'_1(t))\in[0,2\delta]$.
    
    The construction of $\tau'_1$ never moves points upward; that is, we have $f'(\tau'_0(t))-f'(\tau'_1(t))\geq 0$.
    Moreover, it moves points of the violating path to the endpoint $\tau'_0(\ell)$ of the violating path, so it suffices to show that 
    $f'(\tau'_0(t))\leq f'(\tau'_0(\ell))+2\delta$.
    By definition, $\tau'_0=\alpha\circ\tau$, and $f'\circ\alpha$ maps every point exactly $\delta$ up compared to $f$, so it suffices to show that $f(\tau(t))\leq f(\tau(\ell))+2\delta$.

    Let $x = \lca(\tau(\ell), \tau(r))$.
    By Lemma~\ref{lem:subcurves}, the subcurve $[\ell,r]$ of the in-order curve $\tau$ on $T$ is contained in the subtree $T_x$, so $f(\tau(t)) \le f(x)$.
    We now show that $f(x)\leq f(\tau(\ell))+2\delta$.
    Because $(\alpha,\beta)$ is a $\delta$-interleaving, the $2\delta$-ancestors of $\tau(\ell)$ and $\tau(r)$ are equal:
    \[
        \tau(\ell)^{2\delta} = \beta(\alpha(\tau(\ell))) = \beta(\tau'_0(\ell)) = \beta(\tau'_0(r)) = \beta(\alpha(\tau(r))) = \tau(r)^{2\delta},
    \]
    so indeed, the lowest common ancestor $x$ of $\tau(\ell)$ and $\tau(r)$ lies at most $2\delta$ above $\tau(\ell)$.
    In summary, $f(\tau(t)) \le f(x) \le f(\tau(\ell)) + 2\delta$, which is what we needed to show.
\hfill\end{proof}

\subparagraph{Unvisited planted subtrees.}
Lastly, we want to transform our partial in-order curve $\tau'_1$ into an in-order curve $\tau'$ on $T'$ by resolving the $\tau'_1$-unvisited planted subtrees of $T'$.
For this, let $U_1, \ldots, U_k$ be the set of maximal $\tau'_1$-unvisited planted subtrees of $T'$.
We match each unvisited planted subtree $U_i$ to a parameter value $u_i$ that corresponds to a point $\tau'_1(u_i)$.
Suppose that $U_i=T'_{y,v}$ is a maximal unvisited planted subtree.
If $y$ lies interior of an edge, then $\tau'_1$ visits $y$ in exactly one paused interval $[\ell,r]$, and we assign $T'_{y,v}$ to $\ell$.
If on the other hand $y$ is a vertex, then we order the planted subtrees rooted at $y$ based on $\le_h$.
Consider the first one (if any) that is entered by $\tau'_1$ and comes after $T'_{y,v}$, and let $[\ell,r]$ be the paused interval at $y$ just before $\tau'_1$ enters that planted subtree. 
If instead $\tau'_1$ does not enter any planted subtree rooted at $y$ that comes after $T_{y,v}$, let $[\ell,r]$ be the last paused interval at $y$.
In each case, we assign $U_i=T'_{y,v}$ to $u_i\coloneqq\ell$ (the first point of the paused interval).
Note that consecutive $U_i$ may be matched to the same point.
Recall from Lemma~\ref{lem:property-1} that whenever $\ell$ is the first point of a paused interval, we have $f'(\tau'_1(\ell))-f(\tau(\ell))=\delta$, so every unvisited planted subtree $U_i$ is matched to a value $u_i$ with $f'(\tau'_1(u_i))-f(\tau(u_i))=\delta$.

Without loss of generality assume that the unvisited planted subtrees $U_1, \ldots, U_k$ are ordered based on where they appear in an in-order tree walk of $T'$.
We now construct a pair of in-order curves on $T$ and $T'$ respectively.
Specifically, we split both $\tau$ and $\tau'_1$ at the parameters $u_i$, and at $u_i$ insert the constant curve at $\tau(u_i)$ into $\tau$, and insert into $\tau'_1$ an in-order curve on the unvisited tree $U_i$.
We parameterise the resulting curves identically, and call them $\hat\tau$ and $\hat\tau'$, respectively.
Clearly $\hat\tau$ is an in-order curve on $T$.
We show that $\hat\tau'$ is also an in-order curve on $T'$.

\begin{lemma}\label{lem:in-order}
    The curve $\hat\tau'$ is an in-order curve on $T'$.
\end{lemma}
\begin{proof}
    To show that the curve $\hat\tau'$ respects the order of $T'$, we need to show that any $t_1$ and $t_2$ are in the correct order.
    We distinguish four cases based on whether $\hat\tau'(t_1)$ and $\hat\tau'(t_2)$ lie in unvisited subtrees.
    If neither lies in an unvisited subtree, then their order is inherited from $\tau'_1$, which already respected the order of $T'$.
    If $t_1$ lies in an unvisited subtree $U_i$ but $t_2$ does not, then the order is respected by construction of the parameter $u_i$.
    If both lie in the same unvisited subtree, then the order is respected because that subtree is traversed in-order.
    If they lie in different unvisited subtrees $U_i$ and $U_j$, then their order is respected by construction of the parameters $u_i$ and $u_j$.

    Observe that since $\hat\tau'$ visits the entire tree, the unvisited degree of each point $y \in T'$ is $0$.
    Next, fix $y \in T'$.
    We show that $y$ is still visited $\mathrm{deg}(y)+1-0$ times by $\hat\tau'$.
    If $y$ lies interior of a $\tau'_1$-unvisited planted subtree $U_i$, it is visited exactly $\mathrm{deg}(y) + 1$ times by the in-order curve on $U_i$.
    If instead $y$ does not lie interior of a $\tau'_1$-unvisited planted subtree, it was visited $\mathrm{deg}(y)+1-\kappa$ times by $\tau'_1$, where $\kappa$ is the number of $\tau'_1$-unvisited planted subtrees rooted at $y$.
    By construction of $\hat\tau'$, we have increased the number of visits by one for each subtree, resulting in a total of $\mathrm{deg}(y)+1-\kappa+\kappa=\mathrm{deg}(y)+1$ visits of $y$ in $\hat\tau'$.
\hfill\end{proof}

\begin{lemma}\label{lem:delta-matching}
    For all points $t \in [0,1]$, the curve $\hat\tau'$ satisfies $|f(\hat\tau(t)) - f'(\hat\tau'(t))| \le \delta$.
\end{lemma}
\begin{proof}
    Fix a parameter value $t \in [0, 1]$.
    If $t$ is not in one of the newly inserted subcurves of $\hat\tau'$, we know that both $\hat\tau(t) = \tau(t)$ and $\hat\tau'(t) = \tau'(t)$.
    The inequality follows immediately from Lemma~\ref{lem:property-2}.

    Otherwise, $t$ is in one of the newly inserted subcurves of $\hat\tau'$, corresponding to some $\tau'_1$-unvisited planted subtree $U_i$ rooted at $\tau'_1(u_i)$.
    By Lemma~\ref{lem:property-1}, the corresponding curve of $\hat\tau$ is a constant curve at height $f'(\tau'_1(u_i))-\delta$.
    To establish that $|f(\hat\tau(t)) - f'(\hat\tau'(t))| \le \delta$, it is therefore sufficient to show that the depth of $U_i$ is at most $2\delta$.

    Let $c_i$ be a deepest child of $U_i$, and for a contradiction assume that $f'(c_i)<f'(\tau'_1(u_i))-2\delta$.
    Then $(\alpha \circ \beta)(c_i)$ is a strict descendant of $\tau'_1(u_i)$, but in the image of $\alpha$, which means that $U_i$ contains a strict descendant of $u_i$ that is visited by $\tau'_1$, which is a contradiction.
\hfill\end{proof}

\noindent
It follows that the pair $(\hat\tau, \hat\tau')$ is $\delta$-matched.

\begin{lemma}
    $(\hat\tau, \hat\tau')$ is a $\delta$-matched pair of in-order curves for $T$ and $T'$ respectively.
\end{lemma}
\begin{proof}
    By construction and by Lemma~\ref{lem:in-order} respectively, $\hat\tau$ is an in-order curve of $T$ and $\hat\tau'$ is an in-order curve of $T'$.
    Moreover, they are $\delta$-matched by Lemma~\ref{lem:delta-matching}.
\hfill\end{proof}

\begin{corollary}\label{cor:df_leq_dmi}
    $\frechet(f \circ \hat\tau, f' \circ \hat\tau') \le \orddist{}(T, T')$.
\end{corollary}

\noindent
Recall that, for any two in-order curves $\tau$ and $\hat\tau$ of $T$, the 1D curves $f \circ \tau$ and $f \circ \hat\tau$ are equivalent up to reparameterisation (similarly, the 1D curves $f' \circ \tau'$ and $f' \circ \hat\tau'$ are equivalent up to reparameterisation). 
Combining Corollaries~\ref{cor:dmi_leq_df} and \ref{cor:df_leq_dmi} concludes the proof of Theorem~\ref{thm:frechet}.

\subparagraph{Computing the monotone interleaving distance.}
To compute the monotone interleaving distance between two ordered merge trees $(T, f, (\le_h))$ and $(T', f', (\le'_h))$, we can simply do the following.
In linear time, construct a 1D curve $P=f\circ\tau$ based on an arbitrary in-order tree walk $\tau$ of $T$, and a 1D curve $Q=f'\circ\tau'$ based on an arbitrary in-order tree walk $\tau'$ of $T'$.
Now, use any algorithm~\cite{alt95computing, buchin17four} for computing the Fr\'echet distance between 1D curves $P$ and $Q$ and return the result.
\begin{theorem}
    Given two ordered merge trees $T$ and $T'$, there exists an $\tilde\bigO(n^2)$ time algorithm that computes the monotone interleaving distance $\orddist{}(T, T')$ exactly.
\end{theorem}

\section{Categorical Point of View}\label{sec:categorical}
We now show that the Fréchet distance between 1D curves, and therefore also the monotone interleaving distance between ordered merge trees, is an interleaving distance for a specific category.
We first recall the categorical definition of an interleaving~\cite{bubenik2015metrics}, and define the interleaving distance for 1D curves.
Then, we give constructions of a $\delta$-interleaving from a $\delta$-matching, and vice-versa.

\subparagraph{Interleaving distance for a category.}
Denote by $\Rpos$ the poset category based on the standard order on real numbers.
That is, the objects of $\Rpos$ are the real numbers, and there is a morphism from $y$ to $y'$ if $y\leq y'$.
We denote this morphism simply by $(y\leq y')$.
Let $\C$ be a category and $F\from\Rpos\to\C$ be a functor.
For $\delta\geq 0$, the \emph{$\delta$-shift} of $F$ is the functor $F[\delta]\from\Rpos\to\C$ given by $F[\delta](y)=F(y+\delta)$ and $F[\delta](y\leq y')=F(y+\delta\leq y'+\delta)$.
For $\delta\geq 0$, let $\eta^{F,\delta}\from F\To F[\delta]$ be the natural transformation whose component at $y\in\R$ is $F(y\leq y+\delta)$.
For a natural transformation $\varphi\from F\To G$ with components $\varphi_y\from F(y)\to G(y)$, denote by $\varphi[\delta]\from F[\delta]\To G[\delta]$ the natural transformation whose component at $y$ is $\varphi_{y+\delta}\from F[\delta](y)\to G[\delta](y)$.
These natural transformations together with the $\delta$-shift functors, define a functor  $(-)[\delta]\from (\Rpos\to\C)\to(\Rpos\to\C)$ for any $\delta\geq 0$.
Note that $F[\delta+\varepsilon]=F[\delta][\varepsilon]$.
For two functors $F,G\from\Rpos\to\C$, an \emph{$\delta$-interleaving} is a pair of natural transformations $\varphi\from F\To G[\delta]$ and $\psi\from G\To F[\delta]$ such that the following diagrams commute:
\begin{equation}\label{eq:interleaving}\begin{tikzcd}
    F && {F[2\delta]} && {F[\delta]} \\
    & {G[\delta]} && G && {G[2\delta].}
    \arrow["\varphi"', from=1-1, to=2-2]
    \arrow["{\psi[\delta]}"', from=2-2, to=1-3]
    \arrow["{\eta^{F,2\delta}}", from=1-1, to=1-3]
    \arrow["\psi", from=2-4, to=1-5]
    \arrow["{\varphi[\delta]}", from=1-5, to=2-6]
    \arrow["{\eta^{G,2\delta}}"', from=2-4, to=2-6]
\end{tikzcd}\end{equation}
The \emph{interleaving distance} $d_I(F,G)$ is the infimum value $\delta$ for which a $\delta$-interleaving exists.

\subparagraph{Interleaving distance for 1D curves.}
For a function $f\from X\to\R$, the \emph{sublevelset} of $f$ at $y\in\R$ is the set $f\inv((-\infty,y])$.
When~$X$ is a subset of $X'$, we denote the canonical inclusion by $\iota^{X\subseteq X'}\from X\into X'$, or~$\iota$ when it is not ambiguous.
A \emph{1D curve} is a continuous function $P\from[0,1]\to\R$.
Let $\U$ be the category whose objects are closed subsets of $[0,1]$, and whose morphisms are continuous increasing functions.
For a 1D curve $P$, we define a functor $\subinv{P}\from\Rpos\to\U$ given by the sublevelsets $\subinv{P}(y):=P\inv((-\infty,y])$ and canonical inclusions $\subinv{P}(y \leq y'):=\iota$.
We define the \emph{interleaving distance for 1D curves} $P$ and $Q$ as $d_I(\subinv{P},\subinv{Q})$.
In this context, \cref{eq:interleaving} becomes
\begin{equation}\label{eq:interleavingCurves}\begin{tikzcd}
    {\subinv P} && {\subinv P[2\delta]} && {\subinv P[\delta]} \\
    & {\subinv Q[\delta]} && {\subinv Q} && {\subinv Q[2\delta].}
    \arrow["\varphi"', from=1-1, to=2-2]
    \arrow["{\psi[\delta]}"', from=2-2, to=1-3]
    \arrow["\iota", hook, from=1-1, to=1-3]
    \arrow["\psi", from=2-4, to=1-5]
    \arrow["{\varphi[\delta]}", from=1-5, to=2-6]
    \arrow["\iota", hook, from=2-4, to=2-6]
\end{tikzcd}\end{equation}

\subsection{From matchings to interleavings.}
Let $\mu$ be a $\delta$-matching between $P$ and $Q$.
Let $\mu|_X$ be the restriction of $\mu$ to the domain $X$.
We will show that $\varphi_y:=\mu|_{\subinv{P}(y)}$ is a morphism$\from \subinv{P}(y)\to \subinv{Q}[\delta](y)$ in $\U$, and that the morphisms $\varphi_y$ define a natural transformation $\varphi\from\subinv{P}\To\subinv Q[\delta]$.
Symmetric properties apply to $\psi$ resulting from $\psi_y:=\mu\inv|_{\subinv{Q}(y)}$.
Finally, we show that $(\varphi,\psi)$ is a $\delta$-interleaving between $\subinv{P}$ and $\subinv{Q}$.

\begin{lemma}
    $\varphi_y=\mu|_{\subinv{P}(y)}$ is a morphism$\from \subinv P(y)\to \subinv Q[\delta](y)$ in $\U$.
\end{lemma}
\begin{proof}
    Both $\subinv P(y)$ and $\subinv Q[\delta](y)$ are closed and therefore objects of $\U$.
    Because $\mu$ is increasing and continuous, so is the restriction $\mu|_{\subinv P(y)}$.
    By definition $\mu|_{\subinv P(y)}$ has domain $\subinv P(y)$.
    For the codomain, it suffices to show for any $x\in\subinv P(y)=P\inv((-\infty,y])$, the value $\mu(x)$ lies in $\subinv Q[\delta](y)=\subinv Q(y+\delta)=Q\inv((-\infty,y+\delta])$; i.e. $Q(\mu(x))\leq y+\delta$.
    Because $x\in P\inv((-\infty,y])$, we have $P(x)\leq y$.
    Because $\mu$ is a $\delta$-matching we have $\|P(x)-Q(\mu(x))\|\leq\delta$, so $Q(\mu(x))\leq y+\delta$.
\hfill\end{proof}

\begin{lemma}
    The components $\varphi_y$ ($y\in\R$) define a natural transformation $\varphi\from\subinv P\To\subinv Q[\delta]$.
\end{lemma}
\begin{proof}
    It suffices to show that the following square commutes for all $y\leq y'$:
    \begin{equation}\begin{tikzcd}
        {\subinv P(y)} && {\subinv P(y')} \\
        & {\subinv Q[\delta](y)} && {\subinv Q[\delta](y').}
        \arrow["{\varphi_y}", hook, from=1-1, to=2-2]
        \arrow["\iota", hook, from=2-2, to=2-4]
        \arrow["\iota", hook, from=1-1, to=1-3]
        \arrow["{\varphi_{y'}}", hook, from=1-3, to=2-4]
    \end{tikzcd}\end{equation}
    Because the $\iota$'s are canonical inclusions, it is sufficient to show that for any $x\in\subinv P(y)$, we have $\varphi_y(x)=\varphi_{y'}(x)$.
    Indeed, $\varphi_y(x)=\mu|_{\subinv P(y)}(x)=\mu(x)=\mu|_{\subinv P(y')}=\varphi_{y'}(x)$.
\hfill\end{proof}
By exchanging the roles of $P$ and $Q$ and of $\mu$ and $\mu\inv$, symmetric proofs establish that the components $\psi_y=\mu\inv|_{\subinv Q(y)}$ define a natural transformation $\psi\from\subinv Q\To\subinv P[\delta]$.

\begin{lemma}
    $(\varphi,\psi)$ is a $\delta$-interleaving between $\subinv P$ and $\subinv Q$.
\end{lemma}
\begin{proof}
    We first show commutativity for the first triangle of \cref{eq:interleavingCurves}:
    \begin{equation}\begin{tikzcd}
        {\subinv P} && {\subinv P[2\delta].} \\
        & {\subinv Q[\delta]}
        \arrow["\varphi"', from=1-1, to=2-2]
        \arrow["{\psi[\delta]}"', from=2-2, to=1-3]
        \arrow["{\eta^{\subinv P,2\delta}}", from=1-1, to=1-3]
    \end{tikzcd}\end{equation}
    This amounts to showing that $\iota^{\subinv P(y)\subseteq\subinv P(y+2\delta)}=\psi[\delta]_y\circ\varphi_y$ for all $y\in\R$.
    Since $\mu$ is a homeomorphism, $\mu\inv\circ\mu$ is the identity function $\id_{[0,1]}$.
    Therefore, $\psi[\delta]_y\circ\varphi_y$ and $\iota^{\subinv P(y)\subseteq\subinv P(y+2\delta)}$ both are restrictions of the identity function.
    As they have the same domain and codomain, they are equal.
    By a symmetric argument, the second triangle of \cref{eq:interleavingCurves} commutes.
    So, $(\varphi,\psi)$ is a $\delta$-interleaving and $d_I(\subinv P,\subinv Q)\leq d_F(P,Q)$.
\hfill\end{proof}
\begin{corollary}\label{cor:di_leq_df}
    $d_I(\subinv P,\subinv Q)\leq d_F(P,Q)$.
\end{corollary}
    
\subsection{From interleavings to matchings.}
Consider a $\delta$-interleaving $(\varphi,\psi)$ between $\subinv P$ and $\subinv Q$.
Define $y^*:=\max(\Img(P)\cup\Img(Q))+2\delta$ and $\mu:=\varphi_{y^*}\from\subinv P(y^*)\to\subinv Q[\delta](y^*)$.
We show that $\mu$ is a $\delta$-matching between $P$ and $Q$.
\begin{lemma}\label{lem:matching}
    $\mu$ is a matching.
\end{lemma}
\begin{proof}
    $\mu$ is a morphism in $\U$, so it is continuous and increasing.
    We have $\subinv P(y^*)=P\inv((-\infty,y^*])=[0,1]$ and $\subinv Q[\delta](y^*)=Q\inv((-\infty,y^*+\delta])=[0,1]$, so $\mu\from[0,1]\to[0,1]$.
    It remains to show surjectivity.
    Suppose for a contradiction that $\Img(\mu)\neq[0,1]$.
    Then there exists some $z\in[0,1]=\subinv Q(y^*+\delta)$ such that $z\notin\Img(\mu)$.
    Let $x:=\psi_{y^*+\delta}(z)\in[0,1]=\subinv P(y^*)$ and let $z':=\mu(x)\in[0,1]=\subinv Q(y^*+\delta)$. 
    Because $z\notin\Img(\mu)$, we have $z'\neq z$.
    By \cref{eq:interleavingCurves}, we have $\psi_{y^*+\delta}(z')=x$.
    But then $\psi_{y^*+\delta}(z)=\psi_{y^*+\delta}(z')$ for $z\neq z'$, contradicting that $\psi_{y^*+\delta}$ is increasing.
    Hence, $\mu$ is a continuous increasing surjection$\from[0,1]\to[0,1]$.
\hfill\end{proof}

\begin{lemma}
    $\mu$ is a $\delta$-matching between $P$ and $Q$.
\end{lemma}
\begin{proof}
    By Lemma~\ref{lem:matching}, $\mu$ is a matching.
    It remains to show that $\|P(x)-Q(\mu(x))\|\leq\delta$ for all $x\in[0,1]$; that is, $P(x)\leq Q(\mu(x))+\delta$ and $Q(\mu(x))\leq P(x)+\delta$.
    We first show that $P(x)\leq Q(\mu(x))+\delta$.
    Let $y:=P(x)$, then $x\in\subinv P(y)$, so $\mu(x)=\varphi_{y^*}(x)=\varphi_y(x)\in\subinv Q(y+\delta)=Q\inv((-\infty,y+\delta])$, so $Q(\mu(x))\leq y+\delta$.
    Next, we show that $Q(\mu(x))\leq P(x)+\delta$.
    Let $y:=Q(\mu(x))$, then $\mu(x)\in Q\inv((-\infty,y])=\subinv Q(y)$.
    Let $x':=\psi_y(\mu(x))$, then $x'\in\subinv P(y+\delta)=P\inv((-\infty,y+\delta])$ and hence $P(x')\leq y+\delta$.
    By the second triangle of \cref{eq:interleaving}, we have $\varphi_{y+\delta}(x')=\mu(x)$.
    Because $y+\delta\leq y^*$, by naturality of $\varphi$, $\varphi_{y+\delta}$ is a restriction of $\varphi_{y^*}=\mu$.
    Because $\mu$ is increasing, it is injective, so $x=x'$ and hence $P(x)\leq y+\delta=Q(\mu(x))+\delta$.
    Thus,~$\mu$ is a $\delta$-matching between $P$ and $Q$.
\hfill\end{proof}
\begin{corollary}\label{cor:df_leq_di}
    $d_F(P,Q)\leq d_I(\subinv P,\subinv Q)$.
\end{corollary}

\noindent
We combine Corollaries~\ref{cor:di_leq_df} and \ref{cor:df_leq_di} to obtain Theorem~\ref{thm:df_eq_di}.
\begin{theorem}\label{thm:df_eq_di}
    $d_F(P,Q)=d_I(\subinv P,\subinv Q)$.
\end{theorem}

\noindent 
So indeed, the Fréchet distance is an interleaving distance for the category of 1D curves.
As a result, the monotone interleaving distance for ordered merge trees is an interleaving distance for the category of 1D curves that have their endpoints at infinity.

\section{Good Maps and Labellings}\label{sec:variants}
In this section, we describe two alternative ways to define the monotone interleaving distance.
This result is analogous to the regular setting, where the original interleaving distance by Morozov et al.~\cite{morozov2013interleaving} was first redefined by Touli and Wang~\cite{touli2022fpt} in terms of a single $\delta$-shift map with additional properties, and later by Gasparovich et al.~\cite{gasparovich2019intrinsic} in terms of labelled merge trees.
Specifically, Touli and Wang~\cite{touli2022fpt} gave an alternative definition of the interleaving distance in terms of a single $\delta$-shift map with additional requirements.
They call this a \emph{$\delta$-good map}.

\begin{definition2}[Touli and Wang~\cite{touli2022fpt}]\label{def:delta-good}
	Given two merge trees $(T, f)$ and $(T', f')$, a map $\alpha \colon \geom \to \geom'$ is called \emph{$\delta$-good} if and only if
	\begin{enumerate}[({G}1)]
		\item
		$f'(\alpha(x)) = f(x) + \delta$ for all $x \in \geom$,
		
		\item
		if $\alpha(x_1) \succeq \alpha(x_2)$, then $x_1^{2\delta} \succeq x_2^{2\delta}$,
		
		\item
		$|f'(y^F) - f'(y)| \le 2\delta$ for $y \in \geom' \setminus \im(\alpha)$, where $y^F$ is the lowest ancestor of $y$ in~$\im(\alpha)$.
	\end{enumerate}
	The \emph{$\delta$-good interleaving distance} $\intdist{G}((T, f), (T', f'))$ is defined as the smallest $\delta$ such that there exists a $\delta$-good map.
\end{definition2}

\subparagraph{Labelled merge trees.}
For the second alternative definition, we first need some more preliminaries.
A labelled merge tree is a merge tree equipped with a \emph{label map} $\pi$ that maps labels to points in the tree such that each leaf is assigned at least one label.
Specifically, for a given integer $n \ge 0$, we denote by $[n]$ the set $\{1, \ldots, n\}$.
We refer to a map $\pi \colon [n] \to \geom$ as an \emph{$n$-label map} if each leaf in $L(T)$ is assigned at least one label.
Note that $\pi$ may map labels to arbitrary points in~$\geom$ and is not restricted to vertices.
Furthermore, $\pi$ may map different labels to the same point in~$\geom$.
For two points $x_1, x_2 \in \geom$, we write $\fl(x_1, x_2) \coloneqq f(\lca(x_1, x_2))$.
The \emph{induced matrix} of a labelled merge tree captures relevant information about all lowest common ancestors of pairs of labels.
See Figure~\ref{fig:labelled-merge-tree} for an example of a labelled merge tree and its induced matrix.

\begin{figure}[t]
    \centering
	\includegraphics{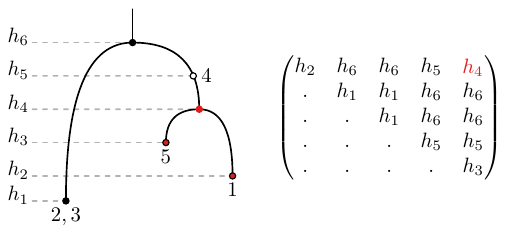}
	\caption{Example of a labelled merge tree and its induced matrix. The red vertices are labelled $1$ and $5$, so the entry $M_{1, 5}$ contains the height of their lowest common ancestor.}
	\label{fig:labelled-merge-tree}
\end{figure}

\begin{definition2}\label{def:labelled-merge-tree}
	An $n$-\emph{labelled merge tree} is a triple $(T, f, \pi)$, with $(T, f)$ a merge tree, and~$\pi \colon [n] \to \geom$ an $n$-label map.
	Its \emph{induced matrix} $M = \imat(T, f, \pi)$ is defined by
	\begin{equation*}
		M_{i, j} = \fl(\pi(i), \pi(j)).
	\end{equation*}
\end{definition2}

\noindent 
We refer to $n$ as the \emph{size} of an $n$-label map or $n$-labelled merge tree.
For simplicity, we usually omit the size, and write labelled merge tree and label map.
When clear from context, we use $T$ to refer to a labelled merge tree.

The field of computational biology uses so-called \emph{cophenetic metrics} to compute distances between \emph{phylogenetic trees}~\cite{cardona2013cophenetic}.
Inspired by these cophenetic metrics, Munch and Stefanou~\cite{munch2019the} define a distance for labelled merge trees of the same size, in terms of their induced matrices.
For a matrix $M$, the $\ell^\infty$-norm is defined as $\norm{M} = \max_{i, j} |M_{i, j}|$.

\begin{definition2}\label{def:label-distance}
	Given two labelled merge trees $(T, f, \pi)$ and $(T', f', \pi')$ with $\pi$ and $\pi'$ two label maps of the same size, and with induced matrices $\imat$ and $\imat'$, the \emph{label distance} is defined as $\labdist((T, f, \pi), (T', f', \pi')) \coloneqq \norm{\imat - \imat'}$.
\end{definition2}

\noindent
As a shorthand we write the label distance as $\labdist(\pi, \pi')$ when the labelled trees $T$ and $T'$ are clear from context.\footnote{Munch and Stefanou~\cite{munch2019the} show that the label distance is an interleaving distance for a specific category. In a follow-up paper~\cite{gasparovich2019intrinsic}, they call the distance the \emph{labelled interleaving distance}. To avoid confusion with the other interleaving distances in our paper, we refer to it as the label distance.}
Given two \emph{unlabelled merge trees}, we can equip them with a pair of equally-sized label maps, which we call a \emph{labelling}.
We can then compare the resulting labelled merge trees using the label distance.

\begin{definition2}\label{def:label-interleaving-distance}
	Given two merge trees $(T, f)$ and $(T', f')$ a pair of $n$-label maps $\pi \colon [n] \to \geom$ and $\pi' \colon [n] \to \geom'$ is called a \emph{$\delta$-labelling}, if
	\begin{equation*}
		\labdist((T, f, \pi), (T', f', \pi')) = \delta.
	\end{equation*}
	The \emph{label interleaving distance} $\intdist{L}((T, f, \pi), (T', f', \pi'))$ is defined as the smallest $\delta$ for which there exists a $\delta$-labelling of size $n = |L(T)| + |L(T')|$. 
\end{definition2}

\noindent
Gasparovich et al.~\cite{gasparovich2019intrinsic} show that the label interleaving distance is equal to the original interleaving distance.
To do so, they first show that for two merge trees $(T, f)$ and $(T', f')$ with interleaving distance $\delta$, there always exists a labelling $(\pi, \pi')$ of size $n = |L(T)| + |L(T')|$ such that the label distance of the labelled merge trees $(T, f, \pi)$ and $(T', f', \pi')$ is $\delta$.
Secondly, they show how to construct a $\delta$-good map from a $\delta$-labelling.
In summary, we have:

\begin{theorem}[\cite{gasparovich2019intrinsic,touli2022fpt}]
	\label{thm:equivalences}
        The distances $\intdist{}$, $\intdist{G}$, and $\intdist{L}$ are equal.
\end{theorem}

\noindent
We say a $\delta$-interleaving, $\delta$-good map, or $\delta$-labelling is \emph{optimal} for two merge trees $T$ and $T'$ if the interleaving distance $\intdist{}(T, T')$ is $\delta$.

\subsection{Monotone variants.}
We now extend the distances $\intdist{G}$ and $\intdist{L}$ to the monotone setting.
Consider ordered merge trees $(T, f, (\le_h))$ and $(T', f', (\le'_h))$.
Recall that $\delta$-shift map is monotone if for all $h \ge 0$, it holds that $x_1 \le_h x_2$ implies $\alpha(x_1) \le_h \alpha(x_2)$ for all $x_1, x_2 \in \ls{h}$.
We now define a \emph{monotone $\delta$-good map} as a $\delta$-good map $\alpha$ that is also monotone.
Moreover, we define a $\delta$-labelling $(\pi, \pi')$ of size $n$ to be \emph{layer-monotone}, or simply \emph{monotone} if for all labels $\ell_1, \ell_2 \in [n]$ it holds that $\pi(\ell_1) \sirel \pi(\ell_2)$ implies $\pi'(\ell_1) \irel' \pi'(\ell_2)$.

\begin{definition2}
    The \emph{monotone $\delta$-good interleaving distance} $\orddist{G}$, and the \emph{monotone label interleaving distance} $\orddist{L}$ is the minimum~$\delta$ that admits a monotone $\delta$-good map, and a monotone $\delta$-labelling, respectively.
\end{definition2}

\noindent
Analogously to Theorem~\ref{thm:equivalences}, we show:
\begin{restatable}{theorem}{monotone}\label{thm:monotone-equivalences}
	The distances $\orddist{}$, $\orddist{G}$, and $\orddist{L}$ are equal.
\end{restatable}

\noindent To prove Theorem~\ref{thm:monotone-equivalences}, we describe three constructions: we construct a monotone $\delta$-good map from a monotone $\delta$-interleaving, a monotone $\delta$-labelling from a monotone $\delta$-good map, and finally a monotone $\delta$-interleaving from a monotone $\delta$-labelling.

\begin{lemma}\label{crl:compatible-good}
	If there exists a monotone $\delta$-interleaving, there exists a monotone $\delta$-good~map.
\end{lemma}
\begin{proof}
    The first construction follows directly from the regular setting considered by Touli and Wang~\cite[Theorem~1]{touli2022fpt}.
    Specifically, they show that the map $\alpha$ in a $\delta$-interleaving $(\alpha, \beta)$ is also a $\delta$-good map.
    If, in addition, $(\alpha, \beta)$ is monotone, this means that the individual map $\alpha$ is also monotone, completing the proof.
\hfill\end{proof}

\subsection{From good maps to labellings.}
Secondly, we construct a monotone $\delta$-labelling from a monotone $\delta$-good map $\alpha$.
To do so, we refine the construction of a $\delta$-labelling from a $\delta$-good map\footnote{We remark that they use a slightly different definition for a $\delta$-good map (see Appendix~\ref{app:good-equivalence}).} by Gasparovich et al.~\cite[Theorem 4.1]{gasparovich2019intrinsic}.
For a point $y \in \geom'$, we use $\lowa{y}$ to denote its lowest ancestor in the image of $\alpha$.
Moreover, for two points $x_1, x_2 \in \geom$, we say $x_1$ is \emph{smaller} than $x_2$, or $x_2$ is \emph{larger} than $x_1$, if $x_1 \irel x_2$.
A point $x$ in a set $X \subseteq T$ is \emph{smallest} (\emph{largest}) if $x$ is smaller (larger) than all other points in $X$.
Similarly, we say a leaf $u_1$ is \emph{smaller} than $u_2$ if $u_1 \ileaf u_2$.
The existing construction is as follows.

\begin{enumerate}[({S}1)]
	\item
	 	For every leaf $u \in L(T)$, add $(u, \alpha(u))$ to an initially empty set $\Pi$.
	\item 
		For every leaf $w \in L(T')$, take an arbitrary point $x \in \alpha^{-1}(w^F)$. Add $(x, w)$ to $\Pi$.
	\item
		Consider an arbitrary ordering $\Pi = \{(x_\ell, y_\ell) \mid \ell \in [n]\}$, and set $\pi(\ell) = x_\ell$, $\pi'(\ell) = y_\ell$.
\end{enumerate}

\noindent The authors show that the resulting labelling is a $\delta$-labelling.
To make sure it is also monotone, we refine (S2) by choosing a specific $x \in \alpha^{-1}(w^F)$.
Intuitively, we first identify all points $\bar{x}$ that lead to a violation of the monotonicity property if we choose $x$ smaller than $\bar{x}$.
Such a violation occurs if $\bar{x}$ is an ancestor of a labelled point whose corresponding labelled point in $\geom'$ is smaller than $w$.
We then take $x$ to be the largest point among the points $\bar{x}$.

\begin{enumerate}[({S}2)]
\item
	For every leaf $w \in L(T')$, sort the set of leaves in $T'_{\lowa{w}}$ by induced leaf-order, denoted $W = \{w_1, \ldots, w_m\}$.
	Define a set of indices $S \subseteq [m]$ as~$S \coloneqq \{k \in [m] \mid \lowa{w_k} \prec \lowa{w}\}$.
	Fix $i \in [m]$ such that $w_i = w$, and define $S_i \subseteq S$ as the set of indices in $S$ strictly smaller than $i$.
	Now consider the set $X = \alpha^{-1}(\lowa{w})$.
	\begin{itemize}
	\item
		If $S_{i}$ is empty, take $x$ to be the smallest point in $X$ and add $(x, w)$ to $\Pi$.
	\item
		If $S_{i}$ is not empty, consider the largest index $\ihat \in S_i$.
		Define $Y$ as the set of strict descendants of $\lowa{w}$ that were labelled in \enumit{(S1)}.
		Consider the following height values:
		\begin{equation}\label{eq:heights}
			\hat{h}_1 \coloneqq \max\{f'(w_k^F) \mid k \in S\},\quad
			\hat{h}_2 \coloneqq \max\{f'(y) \mid y \in Y\}, \quad
			\hat{h} = \max(\hat{h}_1, \hat{h}_2)
		\end{equation}
		Consider the unique ancestor $\hat{w}_\ihat$ of $w_{\ihat}$ at height $\hat{h}$.
		Note that $\hat{w}_\ihat$ is a strict descendant of the point $w^F$ and that it lies in the image of $\alpha$.
		Let $X_{\ihat} \subset X$ be the set of ancestors of points in $\alpha^{-1}(\hat{w}_\ihat)$.
		Take $\hat{x}$ to be the largest point in $X_{\ihat}$ and add $(\hat{x}, w)$ to $\Pi$.
	\end{itemize}
\end{enumerate}

\begin{figure}
	\centering
	\includegraphics{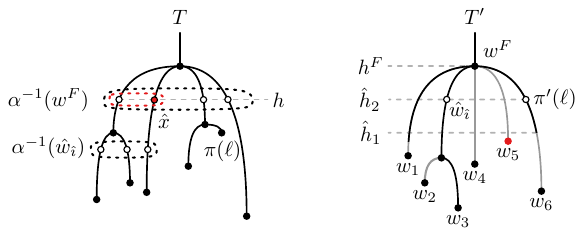}
	\caption{
		The refined step \enumit{(S2)}.
		The grey parts of $T'$ do not lie in the image of $\alpha$.
		We add the pair $(\hat{x}, w_5)$ to $\Pi$.
		In this example, $w = w_5$, $S = \{1, 2, 3, 6\}$, $i = 5$, $S_i = \{1, 2, 3\}$ and $\ihat = 3$.}
	\label{fig:label-construction}
\end{figure}

\noindent See Figure~\ref{fig:label-construction} for an illustration of the modified step \enumit{(S2)}.
We first prove some auxiliary facts about a labelling $(\pi, \pi')$ obtained from the refined construction.
We say a label $\ell$ of the labelling $(\pi, \pi')$ is \emph{constructed} in \enumit{(Si)}, for $i = 1, 2$, if the corresponding pair $(\pi(\ell), \pi'(\ell))$ was added to $\Pi$ in the i-th step of the refined construction.

\begin{lemma}\label{lem:step-s2}
	Consider two labels $\ell_1, \ell_2 \in [n]$ and fix $h' = \max(f'(\pi'(\ell_1)), f'(\pi'(\ell_2)))$.
	If $\pi'(\ell_1)^F \preceq \pi'(\ell_2)^F$ and $\an{\pi'(\ell_1)}{h'} \neq \an{\pi'(\ell_2)}{h'}$, then $\ell_2$ is constructed in \enumit{(S2)}.
\end{lemma}
\begin{proof}
	We do a proof by contradiction.
	Assume label $\ell_2$ is not constructed in \enumit{(S2)}, but rather in \enumit{(S1)}.
	Then, by construction, $\pi'(\ell_2) = \alpha(\pi(\ell_2))$.
	This means that $\pi'(\ell_2)$ lies in the image of $\alpha$, and thus $\pi'(\ell_2)^F = \pi'(\ell_2)$.
	As we assumed $\pi'(\ell_1)^F \preceq \pi'(\ell_2)^F$, it follows that~$\pi'(\ell_1) \preceq \pi'(\ell_1)^F \preceq \pi'(\ell_2)^F = \pi'(\ell_2)$, so $h'= f'(\pi'(\ell_2))$.
	This means that $\pi'(\ell_2)$ is the unique ancestor of $\pi'(\ell_1)$ at height $h'$.
	However, then we get $\an{\pi'(\ell_1)}{h'} = \pi'(\ell_2) = \an{\pi'(\ell_1)}{h'}$, contradicting our assumption.
\hfill\end{proof}

\noindent 
Now, consider a label $\ell$ constructed in \enumit{(S2)}, and let the sets $W$, $X$ and $S$ be the corresponding sets as defined in the above construction.
Define $h, h^F$ and $\hat{h}$ accordingly.
Furthermore, for $k \in S$, denote the unique ancestor of $w_k$ at height $\hat{h}$ by $\hat{w}_k$, and define $X_k \subseteq X$ as the set of ancestors of points in $\alpha^{-1}(\hat{w}_k)$.
For a height value $h \ge 0$ and two sets $X_1, X_2 \subseteq \ls{h}$, we say $X_1 \le_h X_2$ if $x_1 \le_h x_2$ for all $x_1 \in X_1$ and $x_2 \in X_2$.
We now show that the sets $X_i$ as defined in the above construction respect the layer-order.

\begin{restatable}{lemma}{xell}\label{lem:x-ell}
	Consider two indices $i, j \in S$.
	If $i < j$ and $\hat{w}_i \neq \hat{w}_j$, then $X_{i} \le_h X_{j}$.
\end{restatable}
\begin{proof}
	We give a proof by contradiction.
	Assume that not $X_i \le_h X_j$.
	Since $\le_h$ is a total order, that means there exist $x_i \in X_i$ and $x_j \in X_j$ such that $x_j <_h x_i$.
	Since by construction we have $\hat{h} < f'(w_\ell^F)$, it follows that $\hat{h} - \delta < f'(w_\ell^F) - \delta = h$.
	Moreover, we know that as $x_i \in X_i$, there is at least one point $\bar{x}_i \in \alpha^{-1}(\hat{w}_i)$ such that $\bar{x}_i \preceq x_i$.
	Similarly, there also is a point $\bar{x}_j \in \alpha^{-1}(\hat{w}_j)$ such that $\bar{x}_j \preceq x_j$.
	Since $f'(\hat{w}_i) = \hat{h} = f'(\hat{w}_j)$, it follows that $f(\bar{x}_i) = \hat{h} - \delta = f(\bar{x}_j)$.
	Now, we assumed that $x_j <_h x_i$, so from Lemma~\ref{lem:total-order-prop} we obtain that $\bar{x}_j <_{\hat{h}-\delta} \bar{x}_i$.
	By monotonicity of $\alpha$, this means that $\hat{w}_j = \alpha(\bar{x}_j) \le'_{\hat{h}} \alpha(\bar{x}_i) = \hat{w}_i$.
	
	Next, define $h^w = \max(f'(w_i), f'(w_j))$.
	Since $i < j$, we know that $w_i \ileaf' w_j$.
	In other words, we have $\an{w_i}{h^w} \le'_{h^w} \an{w_j}{h^w}$.
	By construction, we know that $f'(w_i) \le \hat{h}$ and $f'(w_j) \le \hat{h}$.
	Together, we thus obtain $h^w \le \hat{h}$.
	Using consistency of $(\le_h)$, it follows that $\hat{w}_i \le_{\hat{h}} \hat{w}_j$.
	We have now shown both $\hat{w}_i \le_{\hat{h}} \hat{w}_j$ and $\hat{w}_j \le_{\hat{h}} \hat{w}_i$.
	By antisymmetry of $\le_h$ we thus get $\hat{w}_i = \hat{w}_j$.
	However, this contradicts our assumption that they are distinct.
\hfill\end{proof}

\noindent 
Before we prove that $(\pi, \pi')$ is monotone, we show one more auxiliary fact.
\begin{lemma}\label{lem:good-property}
	Let $\alpha \colon \geom \to \geom'$ be a $\delta$-shift map.
	Then, for $x \in \geom$ and $h \ge f(x)$, we have~$\alpha(\an{x}{h}) = \an{\alpha(x)}{h + \delta}$.
\end{lemma}
\begin{proof}
	Fix $x \in \geom$.
	First, observe that for $h \ge f(x)$, we have $f'(\alpha(\an{x}{h})) = f(\an{x}{h}) + \delta = h + \delta = f'(\an{\alpha(x)}{h+\delta})$.
	Secondly, by properties of the $\delta$-shift map $\alpha$, we know that as $\an{x}{h} \succeq x$, we have $\alpha(\an{x}{h}) \succeq \alpha(x)$.
	Since $\an{\alpha(x)}{h+\delta}$ is the unique ancestor of $\alpha(x)$ at height $h + \delta$, it follows that $\alpha(\an{x}{h}) = \an{\alpha(x)}{h+\delta}$.
\hfill\end{proof}

\begin{restatable}{lemma}{goodlabel}\label{lem:good-label}
	A $\delta$-good labelling $(\pi, \pi')$ obtained by the refined construction is monotone.
\end{restatable}
\begin{proof}
	Let $n$ be the size of $(\pi, \pi')$.
	To show that $(\pi, \pi')$ is monotone, we need to show that for all labels $\ell_1, \ell_2 \in [n]$, we have $\pi(\ell_1) \sirel \pi(\ell_2)$ implies $\pi'(\ell_1) \irel' \pi'(\ell_2)$.
	Fix $\ell_1, \ell_2 \in [n]$ and set $y_1 = \pi'(\ell_1)^F$ and $y_2 = \pi'(\ell_2)^F$.
	Define height values $h = \max(f(\pi(\ell_1)), f(\pi(\ell_2)))$, $h' = \max(f'(\pi'(\ell_1)), f'(\pi'(\ell_2)))$ and $h^F = \max(f'(y_1), f'(y_2))$.
	Observe that $h' \le h^F$.
	
    \proofsubparagraph{Case 1: $y_1$ and $y_2$ lie in different subtrees.}
	In this case, the points $y_1$ and $y_2$ do not form an ancestor/descendant pair.
	We claim that the points $\an{\pi'(\ell_1)}{h^F}$ and $\an{\pi'(\ell_2)}{h^F}$ are distinct.
	Towards a contradiction, assume they are the same point.
	Without loss of generality, assume $h^F = f'(y_1)$.
	Otherwise, if $h^F = f'(y_2)$, we can use a symmetric argument.
	It follows that $\an{\pi'(\ell_1)}{h^F} = y_1$.
	But then, we get $y_2 \prec \an{\pi'(\ell_2)}{h^F} = \an{\pi'(\ell_1)}{h^F} = y_1$, contradicting our assumption that $y_1$ and $y_2$ do not form an ancestor/descendant pair.
	We know that both $\alpha(\pi(\ell_1)) = y_1$ and $\alpha(\pi(\ell_2)) = y_2$.
	As a result, we have~$h + \delta = h^F$.
	Using Lemma~\ref{lem:good-property}, we obtain the following property for both $i = 1$ and $i = 2$:
	\begin{align}\label{eq:label-property-2}
		\alpha(\an{\pi(\ell_i)}{h}) = \an{\alpha(\pi(\ell_i))}{h+\delta} = \an{y_i}{h^F} = 	\an{\pi'(\ell_i)}{h^F}.
	\end{align}
	Now, assume $\pi(\ell_1) \sirel \pi(\ell_2)$, that is, $\an{\pi(\ell_1)}{h} <_h \an{\pi(\ell_2)}{h}$.
	By monotonicity of $\alpha$, it follows that~$\alpha(\an{\pi(\ell_1)}{h}) \le'_{h+\delta} \alpha(\an{\pi(\ell_2)}{h})$.
	Using Equation~\eqref{eq:label-property-2} and the previous claim that these points are distinct, we obtain $\an{\pi'(\ell_1)}{h^F} <'_{h^F} \an{\pi'(\ell_2)}{h^F}$.
	Since $h' \le h^F$, we can use Lemma~\ref{lem:total-order-prop} to get~$\an{\pi'(\ell_1)}{h'} \le'_{h'} \an{\pi'(\ell_2)}{h'}$.
	In other words, we have $\pi'(i) \lei' \pi'(j)$.
	
    \proofsubparagraph{Case 2: $y_1$ and $y_2$ are the same point.}
	Secondly, we consider the case $\pi'(i)^F = y_1 = y_2 = \pi'(j)^F$.
	Instead of showing that $\pi(\ell_1) \sirel \pi(\ell_2)$ implies $\pi'(\ell_1) \irel' \pi'(\ell_2)$, we show the contrapositive.
	That is, we prove that not $\pi'(\ell_1) \irel' \pi'(\ell_2)$ implies not $\pi(\ell_1) \sirel \pi(\ell_2)$.
	As $\lei$ is a total relation, this statement can be rewritten to $\pi'(\ell_2) \sirel' \pi'(\ell_1)$ implies $\pi(\ell_2) \irel \pi(\ell_1)$.
	For readability and consistency, we prove the symmetric case where the variables $\ell_1$ and $\ell_2$ are swapped.
	So, we prove the following statement: $\pi'(\ell_1) \sirel' \pi'(\ell_2)$ implies $\pi(\ell_1) \irel \pi(\ell_2)$.
	
    Assume $\pi'(\ell_1) \sirel' \pi'(\ell_2)$.
    We denote the point $y_1 = y_2$ by $v^F$.
	By construction, we know that $\alpha(\pi(\ell_1)) = y_1 = v^F$ and $\alpha(\pi(\ell_2)) = y_2 = v^F$.
	As $\alpha$ is a $\delta$-shift map, the points $\pi(\ell_1)$ and $\pi(\ell_2)$ must lie in the same layer, at height $h = f'(v^F) - \delta = h^F - \delta$.
	Consequentially, the unique ancestors of $\pi(\ell_1)$ and $\pi(\ell_2)$ at height $h$ are simply the points $\pi(\ell_1)$ and $\pi(\ell_2)$ themselves.
	So, it suffices to prove $\pi(\ell_1) \le_h \pi(\ell_2)$
	
	As $y_1 = y_2$, both ancestor/descendant relations $y_1 \preceq y_2$ and $y_2 \preceq y_1$ hold.
	Moreover, from the assumption $\pi'(\ell_1) \sirel' \pi'(\ell_2)$, it follows that the points $\an{\pi'(\ell_1)}{h'}$ and $\an{\pi'(\ell_2)}{h'}$ are distinct.
	By Lemma~\ref{lem:step-s2}, we obtain that both labels $\ell_1$ and $\ell_2$ are constructed in \enumit{(S2)}.
	Next, sort the leaves in $T'_{v^F}$ by leaf-order, denoted by $W = \{w_1, \ldots, w_m\}$.
	Sort the points in~$\alpha^{-1}(v^F)$ by $\le_{h}$, denoted $X$.
	Observe that both points $\pi(\ell_1) \in X$ and $\pi(\ell_2) \in X$.
	Define $S = \{k \in [m] \mid w_k^F \prec v^F\}$.
	Lastly, for $i \in [m]$, define $S_{k} = \{k \in S \mid k < i\}$.
	
	As both $\ell_1$ and $\ell_2$ are constructed in \enumit{(S2)}, we know that $\pi'(\ell_1)$ and $\pi'(\ell_2)$ are leaves in~$W$.
	Fix $i, j \in [m]$ such that $w_i = \pi'(\ell_1)$ and $w_j = \pi'(\ell_2)$.
	Since we assumed $\pi'(\ell_1) \sirel' \pi'(\ell_2)$, we know that $i < j$.
	If $S_{i}$ is empty, $\pi(\ell_1)$ is the smallest point in $X$.
	In particular, this means that for $\pi(\ell_2) \in X$, it follows that $\pi(\ell_1) \le_h \pi(\ell_2)$.
	Otherwise, if $S_{i}$ is not empty, we know that as $S_{i} \subseteq S_{j}$, the set $S_j$ is non-empty as well.
	Let $\ihat$ and $\jhat$ be the maxima of $S_{i}$ and $S_{j}$.
	If $\ihat = \jhat$, then $\pi(\ell_1)$ and $\pi(\ell_2)$ were constructed equally, implying $\pi(\ell_1) \le_h \pi(\ell_2)$.
	In the case that $\ihat \neq \jhat$, it holds that $\ihat < i < \jhat < j$, and consequentially that $w_{\ihat} \sileaf' w_{i} \sileaf' w_{\jhat} \sileaf' w_{j}$.
	Define $v \coloneqq \lca(w_{\elll_i}, w_{\elll_j})$.
	We now first prove the following claim.
	
	\begin{claim}\label{clm:lca-leaves}
		We claim that $v = v^F$.
	\end{claim}
	\begin{claimproof}
		Observe that $v^F$ is a common ancestor of $w_{\ihat}$ and $w_{\jhat}$, so $v \preceq v^F$.
		Towards a contradiction, suppose that $v \prec v^F$.
		Since $\ihat \in S$, we know that $w^F_{\ihat} \prec v^F$.
		Now, observe that $w_\ihat$ is a descendant of both $w^F_{\ihat}$.
		As a result, both $v$ and $w^F_\ihat$ lie on the path from $w_\ihat$ to~$v^F$.
		So, either~$w_{\ihat}^F \preceq v$ or $v \prec w_{\ihat}^F \prec v^F$.
		In both cases, there is a point $\hat{v} \in \im(\alpha)$ such that~$v \preceq \hat{v} \prec v^F$: if $w_{\ihat}^F \preceq v$ then by continuity of $\alpha$ we must have $v \in \im(\alpha)$, and if~$v \prec w_{\ihat}^F \prec v^F$ then we can choose $\hat{v} = w_\ihat^F$.
		Recall the separating-subtree property of the leaf-order $\ileaf'$.
		That property tells us that as $w_{\ihat} \sileaf' w_{i} \sileaf' w_{\jhat}$, we have $w_{i} \in T'_v$.
		In particular, we obtain $w_{i} \preceq v \preceq \hat{v} \prec v^F$.
		However, this contradicts our choice of $v^F$, which is the lowest ancestor of $w_i$ in the image of $\alpha$. 
		Hence, we must have $v = v^F$.
	\hfill\end{claimproof}
	
	\noindent Next, define $Y$ as the set of points in $T'_{v^F}$ labelled in step \enumit{(S1)} of the construction, excluding $v^F$ itself.
	Define $\hat{h}_1 = \max\{f'(w^F_k) \mid k \in S\}$, $\hat{h}_2 = \max\{f'(y) \mid y \in Y\}$ and $\hat{h} = \max(\hat{h}_1, \hat{h}_2)$, and note $\hat{h} < h^F$.
	For $k \in S$, define $\hat{w}_k$ as the unique ancestor of $w_k$ at height $\hat{h}$, and define $X_k \subseteq X$ as the set of ancestors of points in $\alpha^{-1}(\hat{w}_k)$.
	We claim that $\hat{w}_\ihat$ and $\hat{w}_\jhat$ are distinct points.	
	Otherwise, if $\hat{w}_\ihat = \hat{w}_\jhat$, the point $\hat{w}_\ihat$ is an ancestor of both $w_\ihat$ and $w_\jhat$, with $f'(\hat{w}_\ihat) = \hat{h}$.
	However, from Claim~\ref{clm:lca-leaves} we know that the lowest common ancestor of $w_\ihat$ and $w_\jhat$ is $v^F$, with height $h^F > \hat{h}$.
	Hence, it follows that $\hat{w}_\ihat$ and $\hat{w}_\jhat$ are indeed distinct points.
	Finally, as $\ihat < \jhat$, we can use Lemma~\ref{lem:x-ell} to get $X_{\ihat} \le_h X_{\jhat}$.
	By construction we know that $\pi(i) \in X_{\ihat}$ and $\pi(j) \in X_{\jhat}$, so we conclude $\pi(i) \le_h \pi(j)$.
 
	\proofsubparagraph{Case 3: one of $y_1$ and $y_2$ is a strict ancestor of the other.}
	Thirdly, consider the case that one of $y_1$ and $y_2$ is a strict ancestor of the other.
	Similar as in Case 2, we show the contrapositive statement: $\pi'(\ell_1) \sirel' \pi'(\ell_2)$ implies $\pi(\ell_1) \irel \pi(\ell_2)$.
	Assume that $\pi'(\ell_1) \sirel' \pi'(\ell_2)$.
	Then we know $\an{\pi'(\ell_1)}{h'} \neq \an{\pi'(\ell_2)}{h'}$.
	This third case consists of two subcases: $y_1 \prec y_2$ and $y_2 \prec y_1$.
	Since these subcases are very similar up to a certain point, we combine them as follows.
	Let $k_1, k_2 \in \{1, 2\}$ such that $y_{k_1} \prec y_{k_2}$.
	By Lemma~\ref{lem:step-s2}, it then follows that label $\ell_{k_2}$ is constructed in \enumit{(S2)}.
	Note that $h^F = f'(\pi'(\ell_{k_2}))$, and that $h = h^F - \delta$.
	
    Again, sort the leaves of $T_{y_{k_2}}$ by leaf order, denoted $W = \{w_1, \ldots, w_m\}$, and set $X = \alpha^{-1}(y_{k_2})$.
	By construction, we know that $\pi(\ell_{k_2}) \in X$.
	Define $S = \{k \in [m] \mid w_k^F \prec y_{k_2}\}$, and $Y$ as the strict descendants of $v^F$ labelled in step \enumit{(S1)}.
	Set $\hat{h}_1 = \max\{f'(w^F_k) \mid k \in S\}$, $\hat{h}_2 = \max\{f'(y) \mid y \in Y\}$ and $\hat{h} = \max(\hat{h}_1, \hat{h}_2)$, and note $\hat{h} < h^F$.
	For $k \in [m]$, define $S_k$ as the set of indices in $S$ smaller than $k$, $\hat{w}_k$ as the unique ancestor of $w_k$ at height $\hat{h}$, and $X_k \subseteq X$ as the set of ancestors of points in $\alpha^{-1}(\hat{w}_k)$.
	Next, consider a leaf $w \in T'_{\pi'(\ell_{k_1})}$.
	Since $w \preceq y_{k_1} \prec y_{k_2}$, we know that $w \in W$ and thus that there exists an $i \in [m]$ such that $w = w_i$.
	Moreover, it follows that $i \in S$.
	We now prove the following claim.
	
	\begin{claim}\label{clm:ell-X}
		The point $\an{\pi(\ell_{k_1})}{h}$ is an element of $X_i$.
	\end{claim}
	\begin{claimproof}
		We first argue that $\pi'(\ell_{k_1})$ is a descendant of $\hat{w}_i$.
		Consider two cases.
		If label $\ell_{k_1}$ is constructed in \enumit{(S1)}, then $\alpha(\pi(\ell_{k_1})) = \pi'(\ell_{k_1})$, and as a result $\pi'(\ell_{k_1}) \in \im(\alpha)$.
		It follows that~$\pi'(\ell_{k_1}) = y_{k_1}$.
		Moreover, as $y_{k_1}$ is a strict descendant of $y_{k_2}$, we know that $\pi'(\ell_{k_1}) \in Y$.
		Together, we get $f'(y_{k_1}) = f'(\pi'(\ell_{k_1})) \le \hat{h}_2 \le \hat{h}$.
		Otherwise, if label $\ell_{k_1}$ is constructed in \enumit{(S2)}, then $\pi'(\ell_{k_1}) \in W$ and $\ell_{k_1} \in S$, implying $f'(y_{k_1}) \le \hat{h}_1 \le \hat{h}$.
		Either way, we have 
		\begin{align}\label{eq:pi-descendant}
			f'(y_{k_1}) \le \hat{h}.
		\end{align}
		Since both $y_{k_1}$ and $\hat{w}_i$ are ancestors of $w_i$, it follows that $y_{k_1} \preceq \hat{w}_i$.
		Consequentially, $\pi'(\ell_{k_1}) \prec y_{k_1}$ implies that $\pi'(\ell_{k_1})$ is a descendant of $\hat{w}_i$.
		In addition, Equation~\eqref{eq:pi-descendant} tells us that $f(\pi(\ell_{k_1})) = f'(y_{k_1}) - \delta \le \hat{h} - \delta$.
		As a result, the point $\an{\pi(\ell_{k_1})}{\hat{h} - \delta}$ is well-defined.
		Using Lemma~\ref{lem:good-property} we obtain $\alpha(\an{\pi(\ell_{k_1})}{\hat{h} - \delta}) = \an{\alpha(\pi(\ell_{k_1}))}{\hat{h}} = \hat{w}_i$.
		In other words, $\an{\pi(\ell_{k_1})}{\hat{h}-\delta} \in \alpha^{-1}(\hat{w}_i)$.
		Recall that $\hat{h} < h^F = h+\delta$.
		Rewriting, we see that $\hat{h} - \delta < h$.
		Hence, as $\an{\pi(\ell_{k_1})}{h}$ is the ancestor of $\an{\pi(\ell_{k_1})}{\hat{h}-\delta} \in \alpha^{-1}(\hat{w}_i)$, we get $\an{\pi(\ell_{k_1})}{h} \in X_{i}$.
	\hfill\end{claimproof}
	
	\noindent Recall that we assumed $\pi'(\ell_1) \sirel' \pi'(\ell_2)$, and that $h = h^F - \delta$.
	Since $h^F = f'(\pi(\ell_{k_2}))$, it follows that $h = f(\pi(\ell_{k_2}))$. 
	We consider the two possible allocations of labels.
 
    \proofsubparagraph{Case 3a: $k_1 = 1$ and $k_2 = 2$.}
	In this case, we have $y_1 \prec y_2$.
	Since label $\ell_2$ is constructed in \enumit{(S2)}, we know that $\pi'(\ell_2)$ is a leaf.
	Fix $j \in [m]$ such that $w_{j} = \pi'(\ell_2)$, and recall that we fixed $i \in [m]$ such that $w_i = w$, where $w \in T'_{\pi'(\ell_1)}$.
	We first argue that $i < j$.
	Set $h^w = \max(f'(w_{i}), f'(w_{j}))$, and observe that $h^w \le h'$.
	As we assumed $\pi'(\ell_1) \sirel' \pi'(\ell_2)$, we have $\an{\pi'(\ell_1)}{h'}  <'_{h'} \an{\pi'(\ell_2)}{h'}$.
	Using Lemma~\ref{lem:total-order-prop}, it follows that $\an{w_i}{h^w} <'_{h^w} \an{w_j}{h^w}$ and thus $w_{i} \sileaf' w_{j}$.
	This implies $i < j$.
	As in addition we have $i \in S$, it follows that $i \in S_{j}$.
	In particular, this means that $S_{j}$ is not empty.
	Let $\jhat$ be the largest index in $S_j$.
	By construction, $\pi(\ell_2)$ is the largest point from the set $X_{\jhat}$.
	From Claim~\ref{clm:ell-X} we know that $\an{\pi(\ell_1)}{h} \in X_i$.
	If $i = \jhat$, we immediately have $\an{\pi(\ell_1)}{h} \in X_\jhat$, and as a result, $\an{\pi(\ell_1)}{h} \le_h \pi(\ell_2)$.
	Otherwise, if $i < \jhat$, then by Lemma~\ref{lem:x-ell} we know that $X_{i} \le_h  X_{\jhat}$.
	In particular, for $\an{\pi(\ell_1)}{h} \in X_i$ and $\pi(\ell_2) \in X_\jhat$, it follows that $\an{\pi(\ell_1)}{h} \le_h \pi(\ell_2)$.
	Either way, we obtain $\pi(\ell_1) \irel \pi(\ell_2)$.
 
    \proofsubparagraph{Case 3b: $k_1 = 2$ and $k_1 = 1$.}
	In this case, we have $y_2 \prec y_1$.
	Since label $\ell_1$ is constructed in \enumit{(S2)}, we know that $\pi'(\ell_1)$ is a leaf.
	Fix $j \in [m]$ such that $w_{j} = \pi'(\ell_1)$, and recall that we fixed $i \in [m]$ such that $w_i = w$, where $w \in T'_{\pi'(\ell_2)}$.
	Using the same argument as in Case~(3a), it now follows that $j < i$.
	If $S_j$ is empty, we know that $\pi(\ell_1)$ is the smallest point in $X$.
	By Claim~\ref{clm:ell-X}, we know that $\an{\pi(\ell_2)}{h} \in X_i \subseteq X$.
	In particular, this means that $\pi(\ell_1) \le_h \an{\pi(\ell_2)}{h}$.
	Otherwise, if $S_j$ is not empty, let $\jhat$ be the largest index in $S_{j}$, and observe that $\jhat < j < i$.
	By Lemma~\ref{lem:x-ell}, it then again follows that $X_{\jhat} \le_h X_{i}$.
	In particular, for $\pi(\ell_1) \in X_j$ and $\an{\pi(\ell_2)}{h} \in X_i$, we get $\pi(\ell_1) \le_h \an{\pi(\ell_2)}{h}$.
	Either way, we obtain $\pi(\ell_1) \lei \pi(\ell_2)$.
\hfill\end{proof}

\subsection{From labellings to interleavings.}
Lastly, we extend the construction of a $\delta$-good map $\alpha \colon \geom \to \geom'$ from a $\delta$-labelling $(\pi, \pi')$ by Gasparovich et al.~\cite{gasparovich2019intrinsic}.
Specifically, for each point $x \in \geom$, they consider a label $\ell$ in the subtree $T_x$.
Note that such a label always exists, as each leaf has at least one label.
Then, they take the unique ancestor $y_\ell$ of $\pi'(\ell)$ at height $f(x) + \delta$, and set $\alpha(x) = y_\ell$.
The authors show~\cite[Theorem 4.1]{gasparovich2019intrinsic} that this point $y_\ell$ is well-defined for any choice of $\ell$, and argue that the resulting map $\alpha$ is a $\delta$-good map.
We use this approach to first construct a $\delta$-shift map $\alpha \colon \geom \to \geom'$, and secondly to construct another $\delta$-shift map~$\beta \colon \geom' \to \geom$.
Together, the construction is as follows:
\begin{enumerate}[\enumit{({S}1)}]
	\item For $x \in \geom$, let $\ell$ be a label in the subtree $T_x$. Set $\alpha(x) = \an{\pi'(\ell)}{f(x) + \delta}$.
	\item For $y \in \geom'$, let $\ell$ be a label in the subtree $T'_y$. Set $\beta(y) = \an{\pi(\ell)}{f'(y) + \delta}$.
\end{enumerate}
\noindent
In the following, we use $S_x \subseteq [n]$ to denote the set of labels in the subtree $T_x$.

\begin{restatable}{lemma}{labelinterleaving}\label{lem:label-interleaving}
	If there is a monotone $\delta$-labelling, then there exists a monotone $\delta$\nobreakdash-interleaving.
\end{restatable}
\begin{proof}
	Let $(\pi, \pi')$ be a monotone $\delta$-labelling.
	Consider the maps $\alpha \colon \geom \to \geom'$ and $\beta \colon \geom' \to \geom$ obtained using the extended construction.	
	To show that the pair $(\alpha, \beta)$ is indeed a monotone $\delta$-interleaving, recall the definition of a $\delta$-interleaving (Definition~\ref{def:interleaving}).
	We briefly argue that $\alpha$ is continuous; continuity of $\beta$ can be proven symmetrically.
	Consider two points $x_1, x_2 \in \geom$ such that $x_1 \prec x_2$.
	Let $i \in S_{x_1}$ be a label in the subtree $T_{x_1}$, so that $\alpha$ maps $x_1$ to the ancestor of $\pi'(i)$ at height $f(x_1) + \delta$.
	The label $i$ is also in the subtree $T_{x_2}$, so $\alpha$ maps $x_2$ to the ancestor of $\pi'(i)$ at height $f(x_2) + \delta$, implying that $\alpha(x_1) \prec \alpha(x_2)$.
	As $\alpha$ is a $\delta$-shift map by construction, it follows that it is continuous.	
	
    \noindent
    To show that $(\alpha, \beta)$ is a $\delta$-interleaving, we argue that conditions \enumit{(C1)}-\enumit{(C4)} hold.
	Conditions \enumit{(C1)} and \enumit{(C3)} of a $\delta$-interleaving follow immediately.
	To see that \enumit{(C2)} holds, we need to show that $\beta(\alpha(x)) = x^{2\delta}$ for all $x \in \geom$.
	Fix $x \in \geom$ and set $h = f(x)$.
	Let $\ell \in S_x$.
	From the construction, we know that $\alpha(x) = \an{\pi'(\ell)}{h+\delta}$.
	Hence, it follows that $\ell \in S_{\alpha(x)}$.
	Consequentially, $\beta(\alpha(x)) = \an{\pi(\ell)}{h+2\delta}$.
	Since $\ell \in S_x$, we know that $\pi(\ell) \preceq x$, and thus $\an{\pi(\ell)}{h+2\delta} = x^{2\delta}$.
	Using a symmetric argument it follows that condition \enumit{(C4)} also holds.
	
	Next, we show that $\alpha$ is monotone.
	Fix $h \ge 0$ and let $x_1, x_2 \in \ls{h}$ such that $x_1 \le_h x_2$.
	To show monotonicity, we need to show $\alpha(x_1) \le'_{h+\delta} \alpha(x_2)$.
	If~$x_1 = x_2$, we trivially get $\alpha(x_1) = \alpha(x_2)$.
	Otherwise, if~$x_1 \neq x_2$, consider two leaves $u_1 \in L(T_{x_1})$ and $u_2 \in L(T_{x_2})$.
	Note that $f(u_1) \le f(x_1) = h$, and similarly $f(u_2) \le h$.
	Defining $h^u = \max(f(u_1), f(u_2))$ we thus see that $h^u \le h$.
	By Lemma~\ref{lem:total-order-prop}, we get $\an{u_1}{h^u} <_{h^u} \an{u_2}{h^u}$.
	As $\pi$ is surjective on the leaves of $T$, there exist labels $\ell_1, \ell_2 \in [m]$ such that $\pi(\ell_1) = u_1$ and $\pi(\ell_2) = u_2$.
	By definition of $\irel$, we have $\pi(\ell_1) \sirel \pi(\ell_2)$.
	Since $(\pi, \pi')$ is monotone, it follows that $\pi'(\ell_1) \irel' \pi'(\ell_2)$.

	Consider the induced matrices $M = M(T, f, \pi)$ and $M' = M(T', f', \pi')$.
	As $(\pi, \pi')$ is a $\delta$-labelling, we know that its label distance is $\delta$.
	In other words, we know that each element of $\mat = \mat(T, T') = |M - M'|$ is bounded by $\delta$.
	In particular, for $k \in \{1, 2\}$, we have:
	\[
	f'(\pi'(\ell_k)) - f(\pi(\ell_k)) = \fl'(\pi'(\ell_k), \pi'(\ell_k)) - \fl(\pi(\ell_k), \pi(\ell_k)) \le \mat_{\ell_k, \ell_k} \le \delta.
	\]
	Rewriting, we have $f'(\pi'(\ell_k)) \le \delta + f(\pi(\ell_k)) \le \delta + h$.
	Since this holds for both $k = 1$ and $k = 2$, it follows that $h' = \max(f'(\pi'(\ell_1)), f'(\pi'(\ell_2))) \le \delta + h$.
	By construction of $\alpha$, we know that $\alpha(x_1)$ is the unique ancestor of $\pi'(\ell_1)$, at height $h + \delta$.
	Consequentially, as $h' \le h + \delta$, the point $\an{\pi'(\ell_1)}{h'}$ lies on the path from $\pi'(\ell_1)$ to $\alpha(x_1)$.
	Similarly, we know that $\an{\pi'(\ell_2)}{h'}$ lies on the path from $\pi'(\ell_2)$ to $\alpha(x_2)$.
	Since we know that $\pi'(\ell_1) \irel' \pi'(\ell_2)$, we obtain $\an{\pi'(\ell_1)}{h'} \le'_{h'} \an{\pi'(\ell_2)}{h'}$.
	By consistency of $(\le'_h)$, we obtain $\alpha(x_1) \le'_{h'} \alpha(x_2)$, showing that $\alpha$ is monotone.
	A symmetric argument shows that $\beta$ is monotone as well.
\hfill\end{proof}

\noindent Theorem~\ref{thm:monotone-equivalences} now follows from Lemmas~\ref{crl:compatible-good}, \ref{lem:good-label}, and \ref{lem:label-interleaving}.

\section{Concluding Remarks}\label{sec:concluding-remarks}
In this paper, we introduced the monotone interleaving distance: an order-preserving distance for a class of ordered merge trees.
The monotone interleaving distance captures more geometric facets of data than the regular interleaving distance, and can hence be a more powerful tool for data that has a geometric interpretation.
We showed that the monotone interleaving distance between two ordered merge trees can be computed in near-quadratic time, by exploiting a relation with the Fr\'{e}chet distance on 1D curves.
Additionally, we demonstrated that the monotone interleaving distance is an interleaving distance in the categorical sense.
Lastly, analogously to the regular setting, we gave three ways to define the monotone interleaving distance: in terms of an interleaving, in terms of a $\delta$-good map, and in terms of a labelling.

We plan to explore various directions for future work.
Firstly, given that the regular interleaving distance for merge trees is NP-hard to compute, and the monotone interleaving distance can be computed in polynomial time for ordered merge trees, it would be interesting to explore which intermediate variants can still be computed in polynomial time.
Specifically, we could relax the total orders on merge tree layers to partial orders.
Furthermore, it would be interesting to explore more connections between variants of the Fr\'{e}chet distance and variants of the (monotone) interleaving distance, which may lead to efficient algorithms or approximations for various types of (monotone) interleaving distances.
Finally, we plan to engineer our results to be applicable to real-world river networks.
As the interleaving distance is a bottleneck distance, an optimal interleaving does not give much information about the parts of the trees that are not the bottleneck.
Therefore, to make the distance more meaningful on real data, we consider developing a \emph{locally adaptive variant} of the (monotone) interleaving distance.

\bibliography{references.bib}

\newpage
\appendix

\section{NP-Hardness of Interleaving Distance}\label{app:np-hardness}
In this section, we show NP-hardness of computing the interleaving distance.
In existing literature it is commonly accepted that computing the interleaving distance is NP-hard \cite{gasparovich2019intrinsic, touli2022fpt}, as a reduction given by Agarwal et al.\ \cite{agarwal2018computing} to prove NP-hardness of the Gromov-Hausdorff distance would also work for the interleaving distance.
For completeness, we slightly alter their proof and show that computing the interleaving distance is indeed NP-hard.

They describe a reduction from the decision problem called balanced partition.
Given a multiset $X = \{a_1, \ldots, a_n\}$ of positive integers and a positive integer $m$, the balanced partition problem asks whether there exists a partition $X_1, \ldots, X_m$ of $X$ such that every set in the partition sums to the same value $\mu \coloneqq 1/m\sum^n_{i=1}a_i$.
The balanced partition problem is strongly NP-complete, meaning that it is NP-complete even if for some constant $c$, it holds that $a_i \le n^c$, for every element $a_i$.
We now describe the construction of two merge trees from an instance $(X, m)$ of the balanced partition problem.

If $m = 1$, the problem is trivial, so assume $m > 1$.
Fix a value $\lambda > 8$.
We first construct a merge tree $(T, f)$ as follows.
For each value $a_i \in X$, we add $a_i$ leaves $\leaf_{i, 1}, \ldots \leaf_{i, a_i}$ at height $0$, all having an incident edge to a vertex $\hat{p}_i$ at height $f(\hat{p}_i) = \lambda$.
We then add $n$ more vertices $p_i$ with $f(p_i) = \lambda + 1$ and edges $(\hat{p}_i, p_i)$.
Moreover, we connect all vertices $p_i$ to a vertex $r$ at height $f(r) = \lambda + 2$.
Lastly, we connect $r$ to a root node at height infinity.
Secondly, we construct a merge tree $(T', f')$ as follows.
For $1 \le j \le m$, we add a group of $\mu$ leaves $w_{j, 1}, \ldots, w_{j, m}$, all at height $1$.
Then, for each of these groups, we connect all its leaves to a vertex $q_j$ at height $f'(q_j) = \lambda+1$.
Finally, connect all vertices $q_j$ to a vertex $r'$ at height $f'(r') = \lambda+3$, and connect $r'$ to a root node at height infinity.
See Figure~\ref{fig:reduction} for an illustration of the two constructed merge trees.

\begin{figure}[b]
	\centering
	\includegraphics{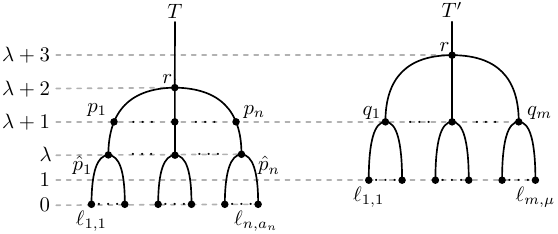}
	\caption{The constructed trees $T$ and $T'$.}
	\label{fig:reduction}
\end{figure}

\begin{lemma}
	If $(X, m)$ is a \emph{yes} instance of the balanced partition problem, then $\intdist{}(T, T') \le 1$.
	Otherwise, we have $\intdist{}(T, T') \ge 3$.
\end{lemma}
\begin{proof}
	Assume $(X, m)$ is a yes instance.
	Then $X$ can be partitioned into $m$ subsets $X_1, \ldots, X_m$ such that each subset sums to $\mu$.
	Consider the merge trees $(T, f)$ and $(T', f')$ obtained by the construction described above.
	We now construct a map $\alpha \colon \geom \to \geom'$ and show it is a $1$-good map.
	Take a set $X_j$, and consider the subtrees rooted at the vertices corresponding to the points in $X_j$.
	We map each leaf in these subtrees to a unique leaf in the subtree $T_{q_j}$.
	Note that this is well-defined: the tree $T_{q_j}$ contains exactly $\mu$ leaves, and by construction there are exactly $\mu$ leaves in $T$ that map to $T_{q_j}$.
	
	Now consider any other point $x \in \geom$, take an arbitrary leaf $\leaf$ in its subtree $T_x$ and set $h = f(x) - f(\leaf)$.
	We map $x$ to the ancestor of $\alpha(\leaf)$ that lies exactly $h$ higher.
	We first argue that the choice of leaf does not matter.
	If $x \in T_{p'_i}$ for some $i$, then $T_x$ contains only a single leaf.
	If $x$ lies on the path from $p'_i$ to $r$ for some $i$, then $\lambda \le f(x) \le \lambda + 2$, and thus $\lambda + 1 \le f'(\alpha(x)) \le \lambda + 3$.
	Since any leaf in $T_x$ maps to the same subtree in $T'$, there is only one candidate at height $\lambda + 1$, namely a point on the path from $q_j$ to $r'$.
	Lastly, if $x \succeq r$, then $f'(\alpha(x)) \ge \lambda + 3$, implying there is only a single candidate for $\alpha(x)$.
	
	Lastly, we show that $\alpha$ is a $1$-good map.
	Condition \enumit{(G1)} is satisfied by definition: every point is mapped to a point in $\geom'$ that lies exactly one higher.
	To see that \enumit{(G2)} holds, consider $x_1$ and $x_2$ such that $y_1 \coloneqq \alpha(x_1) \succeq \alpha(x_2) \eqqcolon y_2$.
	If both $y_1$ and $y_2$ lie on the same edge in a subtree rooted at $q_j$ for some $j$, then by construction $x_1$ and $x_2$ lie on the same edge in a subtree rooted at some $p_i$ with $a_i \in X_j$.
	Otherwise, we must have $f'(y_1) \ge \lambda+1$, and thus $f(x_1^2) \ge \lambda + 2$.
	By construction, this means $x_1^2 \succeq r$ and thus also $x_1^2 \succeq x$ for any point $x$ with $f(x) \le f(x_1^2)$.
	In particular, this holds for $x = x_2^2$.
	Lastly, since $(X, m)$ is a yes instance, every leaf in $T'$ lies in the image of $\alpha$.
	Hence, trivially, the condition \enumit{(G3)} holds as well.
	
	Now, suppose $\delta \coloneqq \intdist{}(T, T') < 3$ and let $\alpha \colon \geom \to \geom'$ be a $\delta$-good map.
	Note that $\delta \ge 1$, since otherwise the leaves of $T$ can not be mapped to points in $T'$.
	We argue that the instance $(X, m)$ is a yes instance.
	Consider the layer $\ls{\delta}' \coloneqq \{y \in \geom' \mid f'(y) = \delta\}$.
	By construction, all leaves of $T$ are mapped to a point in $\ls{\delta}'$.
	We show that the restriction $\chi \colon L(T) \to \ls{\delta}'$ of $\alpha$ on the leaves of $T$ is a bijection.
	
	We first show that $\chi$ is surjective.
	Suppose there is a point $y \in \ls{\delta}'$ such that $y \notin \im(\alpha)$.
	By \enumit{(G3)}, we know that the lowest ancestor $y^F$ in the image of $\alpha$ lies at most $2\delta$ higher than $y$.
	So, since $f'(y) = \delta$, we have $f'(y^F) \le 3\delta < 9$.
	By construction, this means $y^F$ lies on the same edge as $y$, as the higher endpoint $q_j$ for some $j$ lies at height $\lambda + 1 > 9$.
	Let $x \in \geom$ such that $\alpha(x) = y^F$.
	Then $f(x) = f'(y^F) - \delta > f'(y) - \delta = 0$.
	Hence, $x$ can not be a leaf of $T$.
	Let $\leaf \in T_X$ be a leaf in the subtree rooted at $x$.
	By continuity of $\alpha$, it must be that $\alpha(\leaf) \preceq \alpha(x) = y^F$, and moreover $f'(\alpha(\leaf)) = \delta$.
	As there is only a single point at height $\delta$ in the subtree rooted at $y^F$, it follows that $\alpha(\leaf) = y$, which contradicts the assumption that $y \notin \im(\alpha)$.
	To see that $\chi$ is also injective, we make the following observation. 
	By construction, we have $|L(T)| = \sum_{i=1}^n a_i \eqqcolon \overline{a}$.
	Similarly, we have $|\ls{\delta}'| = m \cdot \mu = \overline{a}$.
	As $|L(T)| = |\ls{\delta}'|$, and $\chi$ is surjective, it follows that $\chi$ is injective, and hence bijective.
	
	We can now use this bijection to partition $X$ into $m$ subsets.
	For $i \in [n]$, let $\leaf \in T_{p'_i}$ be an arbitrary leaf.
	Now, let $q_j$ be the unique ancestor of $\alpha(\leaf)$ at height $\lambda+1$.
	We claim that $q_j$ is independent of the choice for $\leaf$.
	Consider two leaves $\leaf_1, \leaf_2 \in T_{p'_i}$, and set $y_1, y_2$ to be their respective images under $\alpha$.
	Note that $\leaf_1^\lambda = p'_i = \leaf_2^\lambda$.
	Then, by continuity, we must have $y \coloneqq y_1^\lambda = y_2^\lambda$.
	Since $1 \le \delta < 3$, we have $\lambda + 1 \le f'(x_1)^\lambda < \lambda + 3$, which means $y$ lies on the edge from $q_j$ to $r$ for some $j$.
	Hence, $y_1, y_2$ lie in the same subtree $T'_{q_j}$, meaning $q_j$ is indeed uniquely determined by $i$.
	so, we can safely add $a_i$ to $X_j$.
	Since $\chi$ is a bijection, there are exactly $\mu$ leaves that map to each subtree in $T'$, meaning that the partition is valid.
\end{proof}

\section{Good Maps}\label{app:good-equivalence}
The original definition of a $\delta$-good map was given by Touli and Wang \cite{touli2022fpt}.
\begin{definition}[Definition 3 from \cite{touli2022fpt}]
	Given two merge trees $(T, f)$ and $(T', f')$, a map~$\alpha \colon \geom \to \geom$ is called $\delta$-good if it is a continuous map such that
	\begin{enumerate}[({T}1)]
		\item
		$f'(\alpha(x)) = f(x) + \delta$ for all $x \in \geom$,
		
		\item
		if $\alpha(x_1) \succeq \alpha(x_2)$, then $x_1^{2\delta} \succeq x_2^{2\delta}$,
		
		\item
		$|f'(y^F) - f'(y)| \le 2\delta$ for $y \in \geom' \setminus \im(\alpha)$, where $y^F$ is the lowest ancestor of $y$ in~$\im(\alpha)$.
	\end{enumerate}
	
\end{definition}

\noindent Gasparovich et al.\ \cite{gasparovich2019intrinsic} use the concept of a $\delta$-good map to show that the label interleaving distance is equal to the interleaving distance.
However, they use a slightly different definition for a $\delta$-good map, and they do not provide a proof that their definition is equivalent.
\begin{definition}[Definition 2.11 from \protect\cite{gasparovich2019intrinsic}]
	Given two merge trees $(T, f)$ and $(T', f')$, a map~$\alpha \colon \geom \to \geom'$ is called $\delta$-good if it is a continuous map such that
	\begin{enumerate}[({G}1)]
		\item for all $x \in \geom$, $f'(\alpha(x)) = f(x) + \delta$,
		\item for all $y \in \im(\alpha)$, $f(x') - f(x) \le 2\delta$, where $x' = \lca(\alpha^{-1}(y))$, for all $x \in \alpha^{-1}(y)$.
		\item for all $y \in \geom' \setminus \im(\alpha)$, $\depth(y) \le 2\delta$, where $\depth(y)$ is defined as $f'(y) - f'(y^L)$, and where $y^L$ is the lowest descendant of $y$.
	\end{enumerate}
\end{definition}

\noindent For completeness, we show these definitions are equivalent.

\begin{lemma}\label{lem:good-equivalence}
	A continuous map $\alpha \colon \geom \to \geom'$ satisfies \enumit{(T1)}-\enumit{(T3)} iff it satisfies \enumit{(G1)}-\enumit{(G3)}.
\end{lemma}
\begin{proof}
	We show both directions.
	
	\begin{description}
	\item[$\Rightarrow$]
		First, suppose $\alpha$ satisfies \enumit{(T1)}-\enumit{(T3)}.
		\begin{enumerate}[({G}1)]
		\item This immediately follows from \enumit{(T1)}.
		
		\item 
			Let $y \in \im(\alpha)$, set $x' = \lca(\alpha^{-1}(y))$ and take $x_1 \in \alpha^{-1}(y)$.
			We consider two cases.
			If $x_1$ is the only point in $\alpha^{-1}(y)$, then $x' = x_1$ and thus $f(x') - f(x_1) = 0 \le 2\delta$.
			Else, let $x_2 \in \alpha^{-1}(y)$ such that $x_1 \neq x_2$.
			Observe that $\alpha(x_1) = \alpha(x_2)$, so both $\alpha(x_1) \succeq \alpha(x_2)$ and $\alpha(x_2) \succeq \alpha(x_1)$.
			By (T2) it then follows that both $x_1^{2\delta} \succeq x_2^{2\delta}$ and $x_2^{2\delta} \succeq x_1^{2\delta}$, and thus $x_1^{2\delta} = x_2^{2\delta}$.
			Since this must hold for any two points $x_1$ and $x_2$ in $\alpha{-1}(y)$, we get $x_1^{2\delta} \succeq x'$.
			In particular, this means $f(x') - f(x_1) \le 2\delta$.
			
		\item
			Let $y \in \geom' \setminus \im(\alpha)$ and towards a contradiction assume $\depth(y) > 2\delta$.
			Then there is $y_L \in \geom'$ such that $f'(y) - f'(y_L) > 2\delta$.
			By continuity of $\alpha$ and tree properties, it must be the case that $y_L \notin \im(\alpha)$, and that $y_L^F \succeq y$.
			But that means we get $|f'(y_L^F) - f'(y_L)| > 2\delta$, contradicting property \enumit{(T3)}.			
		\end{enumerate}
		
	\item[$\Leftarrow$]
		Secondly, suppose $\alpha$ satisfies \enumit{(G1)}-\enumit{(G3)}.
		\begin{enumerate}[({T}1)]
		\item 
			This immediately follows from \enumit{(G1)}.
			
		\item
			Suppose $\alpha(x_1) \succeq \alpha(x_2)$.
			We consider two cases.
			If $x_1 \succeq x_2$, then $x_1^{2\delta} \succeq x_2^{2\delta}$ follows by tree properties.
			Else, we know that $x_1 \neq \lca(x_1, x_2) \neq x_2$.
			Set $y = \lca(\alpha(x_1), \alpha(x_2))$ and observe that $y = \alpha(x_1)$, as $\alpha(x_1) \succeq \alpha(x_2)$.
			Define $h = f'(\alpha(x_1)) - f'(\alpha(x_2)) \ge 0$, so that by continuity of $\alpha$, we have $\alpha(x_2^h) = \alpha(x_1)$.
			This means that both $x_1 \in \alpha^{-1}(y)$ and $x_2^h \in \alpha^{-1}(y)$.
			Set $x' = \lca(\alpha^{-1}(y))$, so that $x' \succeq x_1$ and $x' \succeq x_2^h \succeq x_2$
			By property \enumit{(G2)}, we then obtain $f(x') - f(x_1) \le 2\delta$.
			So, lifting $x_1$ by $2\delta$ yields $x_1^{2\delta} \succeq x' \succeq x_2$.
			As $\alpha(x_1) \succeq \alpha(x_2)$, and by \enumit{(G1)}, we know that $f(x_1) \ge f(x_2)$.
			Hence, we get our desired result: $x_1^{2\delta} \succeq x_2^{2\delta}$.
			
		\item
			Let $y \in \geom' \setminus \im(\alpha)$.
			Then $|f'(y^F) - f'(y)| \le \depth(y^F) \le 2\delta$ by \enumit{(G3)}.
		\end{enumerate}
	\end{description}
	This proves the claim.
\end{proof}
\end{document}